\documentclass[10pt,twocolumn,twoside]{IEEEtran}

\usepackage{cite}
\usepackage{citesort}
\usepackage{graphicx}
\usepackage{amsmath}
\usepackage{times}
\usepackage{subfigure}
\usepackage{latexsym}
\usepackage{graphicx}
\usepackage{bm}
\usepackage{amssymb}
\usepackage[center]{caption2}
\usepackage{stfloats}
\usepackage{cases}
\usepackage{array}
\usepackage{setspace}
\usepackage{fancyhdr}
\usepackage{citesort}

\newlength{\oldparindent}
\newlength{\oldparskip}
\newtheorem{theorem}{Theorem}

\newtheorem{lemma}{Lemma}

\newcommand{\defeq}{\stackrel{\Delta}{=}}

\begin{document}
\title{Tomlinson-Harashima Precoding for Multiuser MIMO Systems with
Quantized CSI Feedback and User Scheduling}

\author{Liang Sun, \emph{Member}, \emph{IEEE} and Matthew R.\ McKay, \emph{Senior Member}, \emph{IEEE}
}
\IEEEaftertitletext{\vspace{-0.75\baselineskip}}

\maketitle


\begin{abstract}
This paper studies the sum rate performance of a low complexity
quantized CSI-based Tomlinson-Harashima (TH) precoding scheme for
downlink multiuser MIMO tansmission, employing greedy user
selection. The asymptotic distribution of the output signal to
interference plus noise ratio of each selected user and the
asymptotic sum rate as the number of users $K$ grows large are
derived by using extreme value theory. For fixed finite signal to
noise ratios and a finite number of transmit antennas $n_T$, we
prove that as $K$ grows large, the proposed approach can achieve the
optimal sum rate scaling of the MIMO broadcast channel. We also
prove that, if we ignore the precoding loss, the average sum rate of
this approach converges to the average sum capacity of the MIMO
broadcast channel. Our results provide insights into the effect of
multiuser interference caused by
quantized CSI on the multiuser diversity gain. 

\end{abstract}


\section{Introduction}
Multiple-input multiple-output (MIMO) communication systems have
received considerable attention in recent years, due to their
ability for providing significantly enhanced spectral efficiency and
link reliability compared with conventional single-antenna systems
\cite{teletar99,Wolniansky98}. In the downlink multiuser MIMO
systems, the spatial multiplexing capability of multiple transmit
antennas can be exploited to efficiently serve multiple users
simultaneously, rather than trying to maximize the capacity of a
single-user link.

The performance of a MIMO communication systems with spatial
multiplexing is severely impaired by the multi-stream interference
due to the simultaneous transmission of parallel data streams. To
reduce interference between the parallel data streams, processing at
the transmitter (precoding) as well as the receiver (equalization)
can be used. Precoding matches the transmission to the channel.
Linear schemes, such are those based on zero-forcing or minimum-mean
squared error (MMSE) criteria \cite{Haustein02,Joham02} or their
regularized variants \cite{Peel05}, are commonly used due to their
low complexity. However, such linear schemes typically incur an
appreciable capacity loss. Nonlinear processing at either the
transmitter or the receiver provides an alternative approach that
offers the potential for performance improvements. Nonlinear
approaches include schemes combining linear precoding with decision
feedback equalization (DFE) \cite{Yang94}, vector
perturbation\cite{Hochwald05}, Tomlinson-Harashima (TH) precoding
\cite{Windpassinger04,Fischer02}, and dirty paper coding (DPC)
\cite{Weingarten06}. Of these, DPC has been shown to be optimal in
terms of achieving the capacity region of the MIMO broadcast
channel\cite{Vishwanath03,Weingarten06}; however it is a highly
nonlinear technique involving joint optimization over a set of
power-constrained covariance matrices, and is therefore generally
deemed too complex for practical implementation \cite{Vishwanath03}.
A reduced complexity sub-optimal DPC scheme which combines linear ZF
with partial DPC, referred to as ZFDPC, was proposed for
single-antenna users in \cite{Caire03}, and generalized to
multiple-antenna users in \cite{Love07}.
Vector perturbation has been proposed for multiuser MIMO channel
model which can achieve rates near sum capacity \cite{Hochwald05}.
It has superior performance to linear precoding techniques as well
as TH precoding \cite{Hochwald05}. However, this method requires the
joint selection of a vector perturbation of the signal to be
transmitted to all the receivers, which is a multi-dimensional
integer-lattice least-squares problem. The optimal solution with an
exhaustive search over all possible integers in the lattice is
prohibitively complex. Although some sub-optimal solutions, such as
the sphere encoder\cite{Damen00}, exist, the complexity is still
much higher than TH precoding.

TH precoding employs modulo arithmetic and has a complexity
comparable to that of linear precoders. It was originally proposed
to combat inter-symbol interference (ISI) in highly dispersive
channels\cite{Harashima72} and can be readily extended to MIMO
channels \cite{Amico08,Windpassinger04}. Although it was shown in
\cite{Hochwald05} that TH precoding does not perform as well as
vector perturbation in general, it can achieve significantly better
performance than linear preprocessing algorithms\cite{Fischer02}.
Thus, it provides a good tradeoff between performance and
complexity. Note that TH precoding is strongly related to DPC. In
fact, it is a suboptimal implementation of the DPC scheme proposed
in \cite{Caire03}.

As with many precoding schemes, the major problem for systems with
TH precoding is the availability of the channel state information
(CSI) at the transmitter. In time division duplex (TDD) systems,
since the channel can be assumed to be reciprocal, the CSI at the
transmitter side can be easily obtained from the channel estimation
during reception. In frequency division duplex (FDD) systems, the
transmitter cannot estimate this information and the CSI has to be
quantized at the receivers and communicated from the receivers to
the transmitter via a feedback channel.

For linear precoding, there have been extensive research results for
MIMO systems with quantized CSI at the transmitter in FDD systems.
\cite{Yoo_limited,Jindal06}. In TDD systems, considering the channel
estimation errors, some robust algorithms can be exploited to
improve the communication performance with THP precoding
\cite{Yik12}. However, as far as we know, there has been little
attention paid to systems employing TH precoding based on quantized
CSI at the transmitter side. Exception include the previous work
\cite{Castro07}, in which TH precoding is designed based on
imperfect CSI where the quantization is performed using scalar
quantization, and the recent work\cite{Israa11}, in which TH
precoding was designed based on the available statistics of the
channel magnitude information (CMI) and the quantized channel
direction information (CDI). In this paper, we will focus on the
implementation of TH precoding in FDD systems.

For many practical precoding schemes, such as ZFDPC and ZF
beamforming (ZFBF), the maximum number of users that can be
supported simultaneously is no larger than the number of transmit
antennas. In practical systems, however, the number of users may be
quite large, and one must select a subset of users to serve at any
given time. Sum rate maximization is a common approach to seek the
subset of supported users. Greedy algorithms are commonly employed,
which can avoid prohibitively large complexity of finding the
optimal subset (see e.g.,
\cite{L.sun_multiuser,Yoo_limited,Love08,Mohammad08}).


In this paper, we design a multiuser spatial TH precoding based on
quantized CSI and a ZF criterion. 
In contrast to \cite{Windpassinger04}, where perfect CSI is assumed
at the transmitter side, here the feedforward filter as well as the
feedback filter are computed at the transmitter only based on the
quantized CDI received at the transmitter side.
For systems with more users than transmit antennas, we propose a low
complexity greedy user scheduling algorithm together with the
quantized CSI-based TH precoding method. We refer to this technique
as G-THP-Q. It is noted that our scheme, based on quantized CSI and
the scheme in \cite{Windpassinger04} based on perfect CSI are both
practical implementations of ZFDPC proposed in \cite{Caire03}. We
present an asymptotic performance analysis of the sum rate (as in
\cite{Yoo_limited,L.sun_multiuser,Love08,Mohammad08}) as the number
of users grows large. In particular, we demonstrate that G-THP-Q
achieves the optimal sum capacity scaling of the MIMO broadcast
channel in this asymptotic regime. In addition, we prove the more
powerful result that if we ignore the precoding loss of TH
precoding, then the difference between the sum rate of G-THP-Q and
the sum capacity of the MIMO broadcast channel converges to zero. We
also establish key insights into the effect of multiuser
interference (MUI) caused by quantized CSI on the multiuser
diversity gain. In particular, we show that, for the perfect CSI
case, whilst the coefficient of the first-order term $O(\log K)$ is
unaffected by the number of iterations $n$ of the proposed greedy
user selection algorithm, the coefficient of second-order term
$O(\log \log K)$ increases linearly with $n_T$ and decreases
linearly with $n$ respectively. In contrast, for the quantized CSI
case, the coefficient of second-order term $O(\log \log K)$
decreases linearly with both $n_T$ and $n$.



\section{System Model}\label{sec:model}
We consider downlink multi-user MIMO systems where TH precoding
\cite{Harashima72} is used at the transmitter for multi-user
interference pre-subtraction. For simplicity, we assume $K$
decentralized users each with single antenna (though the extension
to allow multiple receive antennas is straightforward\footnote{In
\cite{Jindal08}, a receive antenna combining technique called
quantization-based combining is proposed for MU MIMO systems with
multiple-antenna users. For each user, the received signals on
his/her multiple antennas are linearly combined so that the
effective single vector output produced from the original MIMO
channel matrix is closest to one of the codewords in the
quantization codebook, thereby creating an effective single receive
antenna channel vector for each user. This effective channel vector
is quantized and fed back. With this, our proposed TH precoding
method (described in Subsection \ref{subsec:precoder_design}) may
then be performed based on the quantized effective downlink single
receive antenna channels.}), and the transmitter is equipped with
$n_T$ transmit antennas. Let $\mathcal{U} = \left\{1, \ldots,
K\right\} $ denote the set of indices of all $K$ users, and let
$\mathcal{S} \subset \mathcal{U} $ ($L := |\mathcal{S}| \leq n_T$)
be a subset of user indices determined by the scheduler for
transmission. The selected users are indexed by $\mathcal{S} =
\{\mathcal{S}(1),\ldots, \mathcal{S}(L)\}$. The vector $\mathbf{s} =
[s_1, \ldots, s_{L}] \in \mathbb{C}^{L}$ represents the modulated
signal vector, where $s_k$ is the modulated symbol for user
$\mathcal{S}(k)$. (Note that the specific user selection algorithm
will be discussed in Section
\ref{sec:asymptotic_sum_rate_perforamnce}.) Here we assume that in
each of the parallel data streams an $M$-ary square constellation
($M$ is a square number) is employed and the constellation set is
$\mathcal{A}= \big\{ s_I + \sqrt{-1}~s_Q |~s_I,s_Q   \in \pm 1
\sqrt{\frac{3}{2(M-1)}}, \pm 3 \sqrt{\frac{3}{2 (M-1)}}, \ldots, \pm
\left(\sqrt{M}-1 \right)  \sqrt{\frac{3}{2(M-1)}}\big\}$. In
general, the average transmit symbol energy is normalized, i.e.,
$\mathbb{E}\{|s_k|^2\} = 1$. $\mathbf{s}$ is fed to a backward
square matrix $\mathbf{B}$, which must be strictly lower triangular
to allow data precoding in a recursive fashion
\cite{Windpassinger04}. The construction of $\mathbf{B}$ will depend
on the level of CSI of the supported users available at the
transmitter side. In this way, since $\mathbf{B}$ is a function of
the channel matrix, the instantaneous transmit power can be greatly
increased. Thus, TH precoding modulo operation is introduced here to
ensure that the transmit symbols are mapped into the square region $
\mathcal{R} = \{x + \sqrt{-1}~y | x, y\in (-\tau, \tau)\}$, where
$\tau = \sqrt{M}\sqrt{\frac{3}{2(M-1)}}$. The modulo operator
$\text{MOD}_{\tau}(\cdot)$ acts independently over the real and
imaginary parts of its input as follows
\begin{equation}\label{eq:modM}
 \text{MOD}_{\tau} (x) = x -\tau \bigg\lfloor \frac{x+\tau}{2\tau}
\bigg\rfloor ,
\end{equation}
where $\lfloor z \rfloor$
is the largest integer not exceeding $z$. 
Considering the effect of the modulo operation, the channel signals
are given as
\begin{align}
x_1 = s_1, ~~x_k = s_k + d_k - \sum_{l=1}^{k-1}
[\mathbf{B}]_{k,l}x_l, ~ k =2, \ldots, L,\nonumber
\end{align}
where $d_k \in \left\{ 2 \tau (p_I + \sqrt{-1}~p_Q) | ~ p_I , p_Q
\in \mathbb{Z} \right\}$ is properly selected to ensure the real and
imaginary parts of $x_k$ to fall into $\mathcal{R}$
\cite{Windpassinger04}. The constellation of the modified data
symbols $ v_k : = s_k + d_k $ is simply the periodic extension of
the original constellation along the real and imaginary axes.
Equivalently, we have
\begin{align}
\mathbf{v} &= \mathbf{C} \mathbf{x},
\end{align}
where $\mathbf{v} = [v_1,v_2, \cdots, v_{L}]^T$ and $\mathbf{C} :=
\mathbf{B} +\mathbf{I}$. We will make the standard assumption that
the elements of $\mathbf{x}$ are almost uncorrelated and uniformly
distributed over the Voronoi region of the constellation $\mathcal
{R}$. Such a model becomes more precise as $n_T$ increases
\cite[\emph{Theorem} 3.1]{Fischer02}. Compared to transmission of
symbols taken from the constellation $\mathcal {R}$, this leads to a
somewhat increased transmit power quantified by the \emph{precoding
loss} \cite{Windpassinger04}. With $\mathbb{E}\left\{
\mathbf{s}\mathbf{s}^H \right\} = \mathbf{I}$, the covariance of
$\mathbf{x}$ can be accurately approximated as $\mathbf{R_x} =
\frac{M}{M-1}\mathbf{I}$\cite{Fischer02}. Moreover, the induced
shaping loss by the non-Gaussian signaling leads to the fact that
the achievable rate can only be obtained up to $1.53$ dB from the
channel capacity \cite{Wesel98}. However, as indicated in
\cite{Windpassinger04}, the shaping loss can be bridged by
higher-dimensional precoding lattices. A scheme named ``inflated
lattice'' precoding has been proved to be capacity-achieving in
\cite{Uri04}. Thus, following \cite{Windpassinger04}, we will ignore
the shaping loss in this work.

Prior to transmission, a spatial channel pre-equalization is
performed using a feedforward precoding matrix $\mathbf{F} \in
\mathbb{C}^{n_T\times {L}} $. As for $\mathbf{B}$, $\mathbf{F}$ is
also designed based on the level of CSI available to the
transmitter. Throughout this work, we assume equal power allocation
to all supported users. The received signals of all the supported
users can be written in vector form as
\begin{equation}
\mathbf{r}_{\mathcal{S}} =
\sqrt{\frac{P}{\kappa}}\mathbf{H}_{\mathcal{S}}\mathbf{F}\mathbf{x}
+ \mathbf{n}_{\mathcal{S}},
\end{equation}
where $\kappa$ is used for transmit power normalization,
$\mathbf{H}_{\mathcal{S}} = \left[ \mathbf{h}_{\mathcal{S}(1)}^T,
\cdots, \mathbf{h}_{\mathcal{S}(L)}^T \right]^T $ is the equivalent
flat fading channel matrix consisting of all selected users' channel
vectors, and $\mathbf{h}_{\mathcal{S}(k)} \in \mathbb{C}^{n_T} $ is
the channel from the transmitter to user\footnote{The TH precoding
order of the users is ${\mathcal{S}(1)}, {\mathcal{S}(2)}, \ldots,
{\mathcal{S}(L)}$.} ${\mathcal{S}(k)}$. We assume that
$\mathbf{n}_{\mathcal{S}}$ contains the white additive Gaussian
noise at all receivers with covariance
$\mathbf{R}_{\mathbf{n}_{\mathcal{S}}} = \mathbf{I}$, without loss
of generality. Each receiver compensates for the channel gain by
dividing the received signal by a factor $g_k$ prior to the modulo
operation. Let $\mathbf{G} = \text{diag} \left\{g_{1,1}, \ldots,
g_{L,L} \right\}$. Then, the signal vector after channel gain
compensation is
\begin{align}
\mathbf{y}_{\mathcal{S}} & = \mathbf{G} \left (
\sqrt{\frac{P}{\kappa}} \mathbf{H}_{\mathcal{S}}\mathbf{F}\mathbf{x}
+ \mathbf{n}_{\mathcal{S}} \right ). \nonumber
\end{align}

Throughout the paper, we assume (as in
\cite{Yoo_limited,Love08,Mohammad08}) that (i) each user can obtain
perfect knowledge of its own CSI through channel estimation and
feeds back this information to the transmitter via a
rate-constrained feedback link with zero delay, and (ii) the
feedback information can be perfectly received at the transmitter.
In the following subsections, we describe how to determine the TH
precoding matrix $\mathbf{B}$, the feedforward precoding matrix
$\mathbf{F}$ and the compensation matrix $\mathbf{G}$ with quantized
CSI feedback at the transmitter side, based on the selected user set
$\mathcal{S}$.

\section{Precoder Design and User Scheduling Algorithm}

\subsection{Precoder Design with Quantized CSI at the Transmitter}
\label{subsec:precoder_design} In practical systems, perfect CSI is
never available at the transmitter. In a FDD system, the transmitter
obtains CSI of downlink through the limited feedback of $B$ bits by
each receiver. To the best of our knowledge, there has been little
work concerning the design of TH precoding with limited feedback CSI
at the transmitter. Previous work in \cite{Castro07} utilized a
scalar quantization method to quantize each element of the
rank-reduced estimated downlink channels. In \cite{Israa11}, the
random vector quantization (RVQ) method is utilized to quantize the
CDI of the estimated channels. The TH precoding designs in both
papers were based on the MMSE criteria, and depended highly on the
knowledge of the distribution of channel fading. In this work,
following the studies of quantized CSI feedback in
\cite{Jindal06,Yoo_limited}, the channel direction vector is
quantized at each receiver, and the TH precoding is solely designed
based on the information of the quantized channel direction vectors.
Given the quantization codebook with $m$ codewords $\mathcal{W}=
\{\mathbf{w}_1,\cdots, \mathbf{w}_{m}\}$ ($\mathbf{w}_{i} \in
\mathcal{C}^{1\times n_T} $), which is known to both the transmitter
and to all receivers, user $k$ selects the quantized direction
vector as follows:
\begin{equation}
{\hat{\mathbf{h}}_k} = \text{arg} \max_{\mathbf{w}_i
\in\mathbb{W}}\{ |\bar{{\mathbf{h}}}_k \mathbf{w}_i|^2\},
\end{equation}
where $\bar{{\mathbf{h}}}_k = \frac{\mathbf{h}_k}{
\|\mathbf{h}_k\|}$ is the channel direction vector of user $k$. The
corresponding index is fed back to the transmitter via an error and
delay-free feedback channel.

In this work, we use a RVQ codebook, in which the $m$ quantization
vectors are independently and isotropically distributed on the
$n_T$--dimensional complex unit sphere \cite{Jindal06}. 
Using the result in \cite{Jindal06}, for user $k$ we have
\begin{align}\label{eq:channel_decom}
\bar{\mathbf{h}}_k = \hat{\mathbf{h}}_k\cos\theta_k +
\tilde{\mathbf{h}}_k\sin\theta_k,
\end{align}
where $\cos^2\theta_k = |\bar{\mathbf{h}}_k\hat{\mathbf{h}}_k^H|^2$,
whilst $\tilde{\mathbf{h}}_k \in \mathcal{C}^{ 1\times M}$ is a unit
norm vector isotropically distributed in the orthogonal complement
subspace of $\hat{\mathbf{h}}_k$ and is independent of
$\sin\theta_k$. Then $\mathbf{H}_{\mathcal{S}}$ can be written as
\begin{equation}
\mathbf{H}_{\mathcal{S}} = \mathbf{\Gamma}_{\mathcal{S}} \left(
\mathbf{\Phi}_{\mathcal{S}} \hat{\mathbf{H}}_{\mathcal{S}} +
\mathbf{\Omega}_{\mathcal{S}} \tilde{\mathbf{H}}_{\mathcal{S}}
\right),
\end{equation}
where $\mathbf{\Gamma}_{\mathcal{S}} = \text{diag}
\left(\rho_{\mathcal{S}(1)},\ldots ,\rho_{\mathcal{S}(L)} \right) $
with $ \rho_{\mathcal{S}(k)} = \|\mathbf{h}_{\mathcal{S}(k)}\| $,
$\mathbf{\Phi}_{\mathcal{S}}= \text{diag}
\big(\cos\theta_{\mathcal{S}(1)},\ldots ,\cos\theta_{\mathcal{S}(L)}
\big)$ and $\mathbf{\Omega}_{\mathcal{S}} =\text{diag}
\big(\sin\theta_{\mathcal{S}(1)},\ldots
,\sin\theta_{\mathcal{S}(L)}\big) $, $\hat{\mathbf{H}}_{\mathcal{S}}
= \left [\hat{\mathbf{h}}_{\mathcal{S}(1)}^T, \ldots
,\hat{\mathbf{h}}_{\mathcal{S}(L)}^T \right ]^T $ and
$\tilde{\mathbf{H}}_{\mathcal{S}} = \left
[\tilde{\mathbf{h}}_{\mathcal{S}(1)}^T, \cdots
,\tilde{\mathbf{h}}_{\mathcal{S}(L)}^T \right ]^T $. For simplicity
of analysis, in this work we consider the quantization cell
approximation used in \cite{Yoo_limited}, where each quantization
cell is assumed to be a Voronoi region of a spherical cap with
surface area approximately equal to $\frac{1}{m}$ of the total
surface area of the $n_T$-dimensional unit sphere. For a given
codebook $\mathbb{W}$, the actual quantization cell for vector
$\mathbf{w}_i$, $\mathcal{R}_i = \big\{\bar{\mathbf{h}}:
|\bar{\mathbf{h}} \mathbf{w}_i|^2 \geq |\bar{\mathbf{h}}
\mathbf{w}_j|^2, \forall ~ i \neq j \big\} $, is approximated as $
\tilde{\mathcal{R}}_i = \left\{\bar{\mathbf{h}}: |\bar{\mathbf{h}}
\mathbf{w}_i| \geq  1- \delta \right\}$, where $\delta = 2^{-
\frac{B}{n_T -1}}$.

With the quantized CDI at the transmitter, the transmitter obtains
the feedforward precoding matrix $\mathbf{F}$ and feedback matrix
$\mathbf{B} $ through the LQ decomposition of the equivalent channel
matrix $\hat{\mathbf{H}}_{\mathcal{S}} $ as $
\hat{\mathbf{H}}_{\mathcal{S}} =
\hat{\mathbf{R}}_{\mathcal{S}}\hat{\mathbf{Q}}_{\mathcal{S}}, $
where $\hat{\mathbf{R}}_{\mathcal{S}} \in \mathbb{C}^{L \times L}$
is a lower left-triangular matrix and
$\hat{\mathbf{Q}}_{\mathcal{S}} \in \mathbb{C}^{L \times n_T}$ is a
semi-unitary matrix with orthonormal rows which satisfies
$\hat{\mathbf{Q}}_{\mathcal{S}} \hat{\mathbf{Q}}_{\mathcal{S}}^H =
\mathbf{I}$. In addition, we denote $\hat{r}_{i,j}$ as the
$(i,j)$-th element of matrix $\hat{\mathbf{R}}_{\mathcal{S}}$ and
$\hat{\mathbf{q}}_{l}$ as the $l$-th row of matrix
$\hat{\mathbf{Q}}_{\mathcal{S}}$. Then we have $\mathbf{F} =
\hat{\mathbf{Q}}_{\mathcal{S}}^H $ and $\mathbf{B}
=\hat{\mathbf{\Delta}}\hat{\mathbf{H}}_{\mathcal{S}}\mathbf{F}
-\mathbf{I} = \hat{\mathbf{\Delta}} \hat{\mathbf{R}} -\mathbf{I}$
with $\hat{\mathbf{\Delta}}= \text{diag} \left( \hat{r}_{1,1}^{-1},
\ldots, \hat{r}_{L,L}^{-1} \right)$. In addition, the scaling matrix
at the receivers is
\begin{equation}\label{eq:scaling_matrix}
\mathbf{G} = \sqrt{\frac{\kappa}{P}} \left(
\mathbf{\Gamma}_{\mathcal{S}} \mathbf{\Phi}_{\mathcal{S}}
\text{diag}\left\{\hat{\mathbf{R}}_{\mathcal{S}}
\right\}\right)^{-1}.
\end{equation}
According to the transmit power constraint
$\frac{P}{\kappa}\text{Tr}
\{\mathbf{F}\mathbf{R}_{\mathbf{x}}\mathbf{F}^H\}
=\frac{P}{\kappa}\frac{M}{M-1}\text{Tr}\{\mathbf{F}\mathbf{F}^H\} =
\frac{P}{\kappa} \frac{M}{M-1} L = P$, we have $\kappa =
\frac{M}{M-1} L$. After scaling, the effective received vector
$\hat{\mathbf{y}}_{\mathcal{S}}$ can be further written as
\begin{align}\label{eq:Rx_sig_limited}
& \hat{\mathbf{y}}_{\mathcal{S}}  =  \mathbf{G}\left(
\sqrt{\frac{P}{\kappa}}\mathbf{H}_{\mathcal{S}}
\mathbf{F}_{\mathcal{S}} \mathbf{x}
+\mathbf{n}_{\mathcal{S}}\right)\nonumber\\
&= \mathbf{G} \sqrt{\frac{P}{\kappa}}\mathbf{\Gamma}_{\mathcal{S}}
\left( \mathbf{\Phi}_{\mathcal{S}} \hat{\mathbf{H}}_{\mathcal{S}} +
\mathbf{\Omega}_{\mathcal{S}} \tilde{\mathbf{H}}_{\mathcal{S}}
\right) \mathbf{F} \mathbf{x} + \mathbf{G}\mathbf{n}_{\mathcal{S}}\nonumber\\
&= \mathbf{v} + \left( \mathbf{\Phi}_{\mathcal{S}}
\text{diag}\left\{\hat{\mathbf{R}}_{\mathcal{S}}
\right\}\right)^{-1}\mathbf{\Omega}_{\mathcal{S}}
\tilde{\mathbf{H}}\hat{\mathbf{Q}}^H\mathbf{x} \nonumber \\ &
\hspace{1cm }+ \sqrt{\frac{\kappa}{P}}\left(
\mathbf{\Gamma}_{\mathcal{S}}
\mathbf{\Phi}_{\mathcal{S}}\text{diag}\left\{\hat{\mathbf{R}}_{\mathcal{S}}
\right\}\right)^{-1}\mathbf{n}_{\mathcal{S}},
\end{align}
where we have used the relationship $\mathbf{v} =
\left(\text{diag}\left\{\hat{\mathbf{R}}_{\mathcal{S}}
\right\}\right)^{-1}\hat{\mathbf{R}}_{\mathcal{S}} \mathbf{x}$. The
first term is the useful signal for all the selected users, the
second term is the interference caused by quantized CSI, whilst the
last term is the effective noise after processing. At the receivers,
each symbol in $\mathbf{y}_{\mathcal{S}}$ is first modulo reduced
into the boundary region of the signal constellation $\mathcal{A}$.
A slicer of the original constellation will follow the modulo
operation to detect the received signals. According to
(\ref{eq:Rx_sig_limited}), the signal to interference plus noise
ratio (SINR) $\gamma_{\mathcal{S}(k)}$ for user $\mathcal{S}(k)$
($k=1, \ldots,{L}$) can be written as
\begin{align}\label{eq:SINR_limited}
&\gamma_{\mathcal{S}(k)} \nonumber \\ &= \frac{1}{
\frac{\sin^2\theta_{\mathcal{S}(k)}}{|\hat{r}_{k,k}|^2\cos^2\theta_{\mathcal{S}(k)}}
\|\tilde{\mathbf{h}}_{\mathcal{S}(k)}
\hat{\mathbf{Q}}_{\mathcal{S}}^H \|^2+
\frac{\kappa}{P}\frac{1}{\rho_{\mathcal{S}(k)}^2|\hat{r}_{k,k}|^2 \cos^2\theta_{\mathcal{S}(k)}} } \nonumber\\
&= \frac{\frac{P}{\kappa} \rho_{\mathcal{S}(k)}^2|\hat{r}_{k,k}|^2
\cos^2\theta_{\mathcal{S}(k)}}{ \frac{P}{\kappa}
\rho_{\mathcal{S}(k)}^2 \| \tilde{\mathbf{h}}_{\mathcal{S}(k)}
\hat{\mathbf{Q}}_{\mathcal{S}(k)}^H \|^2
\sin^2\theta_{\mathcal{S}(k)}+1 }.
\end{align}

In the next subsection, we present a greedy scheduling algorithm
which is combined with the proposed quantized CSI-based TH precoding
to determine the supported user set $\mathcal{S}$. Henceforth, this
strategy will be termed G-THP-Q.

\subsection{User Scheduling}\label{subsec:scheduler}

In the above subsection, the transceiver structures are obtained
based on the selected user set $\mathcal{S}$, which is determined by
the scheduler at the transmitter. Given the optimal user set
$\mathcal{S}$, with equal power allocation to each user, the sum
rate is
\begin{equation}\label{eq:throughtput}
R_{\text{G-THP-Q}} = \sum_{i=1}^{L} \log_2 \left(1 +
\frac{P}{\kappa} \gamma_{\mathcal{S}(k)}\right).
\end{equation}

In the following, we consider the problem of determining the
\emph{optimal} user set to maximize the sum rate given by
(\ref{eq:throughtput}). When $K \gg n_T$, to find the optimal user
set $\mathcal{S}$, for each $L \leq n_T$ an \emph{exhaustive search}
must be applied over all possible sets of $L$ user channels. In
addition, all permutations of a given user set must be considered
due to the fact with TH precoding that different orderings of a
given set of user channels yield different sum rates. For large
values of $K$ the complexity associated with this exhaustive search
is prohibitive in practice \cite{L.sun_multiuser,Yoo06}. To reduce
the complexity of user scheduling we adopt the greedy method which
has been extensively used in the literature
\cite{L.sun_multiuser,Yoo_limited,Love08,Mohammad08} and in
practical systems \cite{Farooq09}. Besides the CDI feedback, each
user also feeds back its channel magnitude $\rho_k^2 =
\|\mathbf{h}_k\|^2$, which will be used for scheduling. The proposed
user selection method iteratively selects a user by searching for a
set of users, based on the quantized CSI. Let $\mathcal{U}_n$ denote
the \emph{candidate set} at the $n$-th iteration. This set contains
the indices of all users which have not been selected previously.
Also, let $\mathcal{S}_n= \{\mathcal{S}(1), \ldots,
\mathcal{S}(n)\}$ denote the set of indices of the selected users
after the $n$-th iteration. The selection algorithm works as
follows. \vspace*{3mm} \hrule \vspace*{2mm} \textbf{G-THP-Q
(Algorithm 1)} \vspace*{2mm} \hrule \vspace*{2mm}
\begin{enumerate}
\item
\textbf{Initialization:}\\
Set $n=1$ and $\mathcal{U}_1 = \left\{1,2,\ldots,K\right\}$. \\
Let
\begin{equation}\label{eq:gamma_1}
\gamma_{k}(1) = \frac{\frac{P}{\kappa} \rho_{k}^2 \cos^2\theta_{k}}{
\frac{P}{\kappa} \rho_{k}^2 \sin^2\theta_{k}+1 },~~ k \in
\mathcal{U}_1
\end{equation}
The transmitter selects the first user as follows:
\begin{eqnarray}\label{eq:selection1}
\mathcal{S} (1) = \textrm{arg} \max_{k \in \mathcal{U}_1}
~\gamma_{k}(1) \, .
\end{eqnarray}
Set $\mathcal{S}_1 = \left\{\mathcal{S}(1) \right\}$, and define
$\hat{\mathbf{q}}_1 = \hat{\mathbf{h}}_{\mathcal{S}(1)}$.
\item
\textbf{While} $n \leq n_T$, $n \leftarrow n+1$. \\
Candidate set is $\mathcal{U}_n = \mathcal{U}_{n-1}\backslash
\mathcal{S}_{n-1} $. For each user $k \in \mathcal{U}_n$, denote
\begin{eqnarray}
\label{eq:Gram-Schmidt0}
\xi_{i} &=& \hat{\mathbf{h}}_{k}\hat{\mathbf{q}}_{i}^H,~~ i=1, \ldots, n-1\\
\label{eq:Gram-Schmidt1} \boldsymbol{\xi}_{k} &=&
\hat{\mathbf{h}}_{k} - \sum_{i=1}^{n-1}\xi_{i}\hat{\mathbf{q}}_{i}\\
\label{eq:gamma_k} \gamma_{k}(n) &=&
\frac{\frac{P}{\kappa}\rho^2_{k}
\parallel \boldsymbol{\xi}_{k}\parallel^2 \cos^2\theta_k}{\frac{P}{\kappa}\rho^2_{k} \sin^2\theta_k +1}.
\end{eqnarray}
Select the $n$-th supported user as follows:
\begin{eqnarray}\label{eq:selection2}
\mathcal{S}(n) &=& \text{arg} \max_{ k \in \mathcal{U}_n}
\gamma_{k}(n) \, .
\end{eqnarray}
Set $\mathcal{S}_n = \mathcal{S}_{n-1} \cup \{\mathcal{S}(n)\}$ and
\begin{eqnarray}\label{eq:Gram-Schmidt2}
\hat{\mathbf{q}}_n = \frac{\boldsymbol{\xi}_{\mathcal{S}(n)}}{
\parallel\boldsymbol{\xi}_{\mathcal{S}(n)}\parallel}.
\end{eqnarray}
\item
Let $\mathbf{F}= \hat{\mathbf{Q}}_{\mathcal{S}}^H =
\left[\hat{\mathbf{q}}_1^H,\cdots, \hat{\mathbf{q}}_n^H \right]$,
$\mathbf{G}$ is obtained by using (\ref{eq:scaling_matrix}). The
transmitter broadcasts to all users the indices of the selected
users; then performs TH precoding as discussed previously.
\end{enumerate}
\hrule  \vspace*{5mm} Note that it is obvious that $L = n_T$ users
are determined at the end of this user scheduling algorithm. In this
case, since $\hat{\mathbf{Q}}_{\mathcal{S}}$ becomes a unitary
matrix, if user $k \in \mathcal{U}_n$ is scheduled, the term $\|
\tilde{\mathbf{h}}_k \hat{\mathbf{Q}}_{\mathcal{S}}^H \|^2$ in the
denominator of (\ref{eq:SINR_limited}) for the $n$-th scheduled user
becomes a constant, i.e., $\| \tilde{\mathbf{h}}_k
\hat{\mathbf{Q}}_{\mathcal{S}}^H \|^2 = 1$. In addition, the power
allocated to each user is equal to $\frac{P}{n_T}$. Thus
(\ref{eq:gamma_k}) follows.

The following important relations can be observed. According to the
LQ decomposition of the ordered rows of
$\hat{\mathbf{H}}_{\mathcal{S}}$, as described by
(\ref{eq:Gram-Schmidt0}), (\ref{eq:Gram-Schmidt1}) and
(\ref{eq:Gram-Schmidt2}),
\begin{eqnarray}\label{eq:QR2}
\hat{\mathbf{h}}_{\mathcal{S}(n)}= ~
(\hat{\mathbf{h}}_{\mathcal{S}(n)}~\hat{\mathbf{q}}_n^H
)\hat{\mathbf{q}}_n + \sum_{j=1}^{n-1}(
\hat{\mathbf{h}}_{\mathcal{S}(n)} \hat{\mathbf{q}}_j^H )
\hat{\mathbf{q}}_j,
\end{eqnarray}
and $\hat{r}_{n,j} = \hat{\mathbf{h}}_{\mathcal{S}(n)}
\hat{\mathbf{q}}_j^H$, for $ j \leq n$. With
(\ref{eq:Gram-Schmidt1}),
\begin{eqnarray}
|\hat{r}_{n,n}|^2 = |\hat{\mathbf{h}}_{\mathcal{S}(n)}
\hat{\mathbf{q}}_n^H |^2 =
\parallel \boldsymbol{\xi}_{\mathcal{S}(n)} \parallel^2.
\end{eqnarray}
In addition, since $\|\hat{\mathbf{h}}_{\mathcal{S}(n)}\|^2 = 1$ and
$\hat{\mathbf{q}}_i, ~i=1,\ldots, n_T$ are orthonormal, it can be
easily shown that
\begin{eqnarray}\label{eq:sum_l}
\sum_{j=1}^{n} |\hat{r}_{n,j}|^2 = 1, ~~~~\text{for}~
n=1,2,\ldots,n_T.
\end{eqnarray}
These relations will be useful for the subsequent analysis.

\section{Sum Rate Analysis}\label{sec:asymptotic_sum_rate_perforamnce}

In this section, we investigate the sum rate achieved by the
proposed quantized CSI-based TH precoding scheme combined with the
greedy user scheduling algorithm described in Subsection
\ref{subsec:scheduler}. Besides the assumptions made in Section
\ref{sec:model}, we further assume the channels of all users are
subject to uncorrelated Rayleigh fading and, for simplicity, as in
\cite{Dimic05,Love08,Yoo_limited} all users are homogeneous and
experience statistically independent fading. We focus on
establishing asymptotic results as $K \to \infty$, whilst keeping
SNR and $n_T$ finite.


To analyze the sum rate of the system, we require the distribution
of the output SINR $\gamma_{\mathcal{S}(n)}$ of the selected user at
the $n$-th iteration of the user selection algorithm. The SINR for
\emph{arbitrary} user $k$ selected from the candidate set
$\mathcal{U}_n$ at the $n$-th iteration of \textbf{Algorithm 1} is
given by
\begin{equation}\label{eq:SINR_limited_sel}
\gamma_{k}(n)= \frac{\phi \rho_k^2 \omega_k(n) \cos^2\theta_k}{ \phi
\rho_k^2\sin^2\theta_k+1 }, ~~ k \in \mathcal{U}_n,
\end{equation}
with $\omega_k(n) = \parallel \boldsymbol{\xi}_{k}\parallel^2$ given
by (\ref{eq:Gram-Schmidt1}) and $\phi = \frac{(M-1) P}{M n_T}$.
Thus, let us first determine the \emph{common} distribution of
$\gamma_{k}(n)$ with $n = 1, \ldots, n_T$.

Starting with $n = 1$, $\gamma_{k}(1)$ in (\ref{eq:gamma_1}), $k=1,
2\ldots, K$, are independent and identically
distributed (i.i.d.) random variables 
whose common cumulative distribution function (c.d.f.) has been
obtained in\cite{Yoo_limited} in closed-form as\footnote{The
expression for $F_{\gamma(n)} (x)$ for $x < \frac{1}{\delta} -1$ is
more involved. Since only the \emph{tail} behavior (large $x$) of
$F_{\gamma(n)}(x)$ is used for analysis, throughout this paper it is
sufficient to focus on the case $~x \geq \frac{1}{\delta} -1$.}
\begin{equation}\label{eq:cdf_exact1}
F_{\gamma(1)} (x) = 1- \frac{2^B
e^{-\frac{x}{\phi}}}{(1+x)^{n_T-1}},~~~~x \geq \frac{1}{\delta} -1.
\end{equation}

For $n\geq 2$, it is more challenging to obtain the distribution of
$\gamma_{k}(n)$, $k \in \mathcal{U}_n$. Particularly, after the user
selection in (\ref{eq:selection2}) in Step $2$ at the previous
iteration (i.e., the $\left(n-1\right)$-th), the exact distribution
of the channel vectors in $\mathcal{U}_n$ is different from the
distributions of the channel vectors in $\mathcal{U}_l$, $ l \leq
n-1$. More specifically, for $n\geq 2$, the channels for users in
the candidate set $\mathcal{U}_n$ no longer behave statistically as
uncorrelated complex Gaussian vectors. We see from
(\ref{eq:SINR_limited_sel}) that, compared with $\gamma_{k}(1)$,
$\gamma_{k}(n)$ ($n\geq 2$) involves an additional variable
$\omega_k(n) = \parallel \boldsymbol{\xi}_{k}\parallel^2$. For the
reasons stated above, the exact distributions of $\omega_k(n)$,
$\rho_k^2$, $\cos^2\theta_k$ and $ \sin^2\theta_k$ for $k\in
\mathcal{U}_n, n\geq 2$ are currently unknown and appear difficult
to derive analytically. Thus, we can continue the analysis by
considering the ``large-user'' regime. In particular, using similar
techniques as in \cite{Wang_report07,Love08,L.sun_multiuser} we can
strictly prove that, when the number of users in the candidate set
$\mathcal{U}_n$ is large, removing one user from $\mathcal{U}_n$ has
negligible impact on the statistical properties of the remaining
users' channel vectors.

\begin{lemma}\label{lemma:iid_property}
At the $n$-th iteration of \textbf{Algorithm 1}, $ 2\leq n \leq M$,
conditioned on the previously selected channel vectors
$\mathbf{h}_{\mathcal{S}(1)}, \ldots,
\mathbf{h}_{\mathcal{S}(n-1)}$, the channel vectors in
$\mathcal{U}_n$ are i.i.d. Furthermore, as the size of the candidate
user set $\mathcal{U}_n$ grows large (i.e., $\lim_{K \to \infty}
|\mathcal{U}_n| = \infty$), conditioned on the previously selected
channels, the distribution of the channel vector of each user in
$\mathcal{U}_n$ converges to the distribution of uncorrelated
complex Gaussian vector.
\end{lemma}

Equipped with \emph{Lemma} \ref{lemma:iid_property}, at the $n$-th
iteration, from the point of view of the users in $\mathcal{U}_n$,
the channel vectors of the selected users in the previous iterations
(i.e., $\mathbf{h}_{\mathcal{S}(n)}, \ldots,
\mathbf{h}_{\mathcal{S}(n)}$) appear to be \emph{randomly} selected.
Thus, the orthonormal basis $\hat{\mathbf{q}}_1, \ldots,
\hat{\mathbf{q}}_{n-1}$ (generated from
$\hat{\mathbf{h}}_{\mathcal{S}(n)}, \ldots,
\hat{\mathbf{h}}_{\mathcal{S}(n)}$) appears independent of the
channel vectors of the users in $\mathcal{U}_n$. This greatly
simplifies the following analysis.

Firstly, we will derive the distribution of the random variable $
\omega_k(n) $ for $k \in \mathcal{U}_n$ $n \geq 2$, which is given
in the following lemma.
\begin{lemma}\label{lemma:pdf_omega_n}
Let $ k \in \mathcal{U}_{n}$, $n \in \{ 2, \ldots, M \}$. For
sufficiently large $K$, $\omega_k(n) =\parallel
\boldsymbol{\xi}_{k}\parallel^2$ follows a beta distribution with
shape parameters $(n_T-n+1)$ and $(n-1)$, which is denoted as
$\omega_k(n) \sim \mathrm{Beta}(n_T-n+1,n-1)$ whose probability
density function (p.d.f.) is given as follows:
\begin{eqnarray}\label{eq:pdf_omega_n}
f_{\omega(n)} (x) = \frac{1}{\beta(n_T-n+1,n-1)}x^{n_T-n}(1-
x)^{n-2},
\end{eqnarray}
where $\beta(a,b) = \int_{0}^{1} t^{a-1} (1-t)^{b-1} {\rm d}t$ is
the beta function \cite{Gradshteyn2000}.
\end{lemma}
\begin{proof}
See Appendix \ref{proof:lemma_pdf_omega_n}.
\end{proof}

Note that $\gamma_{k}(n)$ involves random variables $\rho_k^2$,
$\cos^2\theta_k$ and $\sin^2\theta_k$. In the following, we recall
the joint distribution of $\rho_k^2 \cos^2\theta_k$ and $\rho_k^2
\sin^2\theta_k$ which was obtained in \cite{Yoo_limited}.
\begin{lemma}\label{lemma:S_I_distribution}
Under the quantization cell approximation of $ \tilde{\mathcal{R}}_i
$ shown in Section \ref{subsec:precoder_design}, the joint
distribution of $\big(\rho_k^2 \cos^2\theta_k, \rho_k^2
\sin^2\theta_k \big)$ is the same as that of $(X + (1-\delta)Y,
\delta Y )$, where $X$ is gamma-distributed with shape $1$ and scale
$1$ and $Y$ is gamma-distributed with shape $(M-1)$ and scale $1$.
These are denoted as $X \sim \mathrm{Gamma}(1, 1)$ and $Y \sim
\mathrm{Gamma}(M - 1, 1)$ respectively.
\end{lemma}

Denote the c.d.f. of $\gamma_{k}(n)$ for $k \in \mathcal{U}_{n}$ as
$F_{\gamma(n)}(x)$. Equipped with \emph{Lemma}
\ref{lemma:pdf_omega_n} and \emph{Lemma}
\ref{lemma:S_I_distribution}, we can obtain $F_{\gamma(n)}(x)$ in
the following lemma.
\begin{lemma}\label{lemma:SINR_UE_selection}
The c.d.f. $F_{\gamma(n)}(x)$ for $n = 2, \ldots, n_T$ is given by
\begin{eqnarray}\label{eq:cdf_exact2}
&\hspace{-4cm} F_{\gamma(n)}(x) = 1- a_n ~ x^{n_T-n+1}\nonumber\\
&  \times  V(n-1;-n_T+2;-n_T+1;\frac{1}{\phi}; x),~~~x \geq
\frac{1}{\delta} -1,~~
\end{eqnarray}
where $a_n = \frac{2^B}{\beta(n_T-n+1,n-1)}$, 
\begin{align}
V(m_1;m_2;m_3;\mu; x) & = \int_{x}^{\infty} e^{-\mu t}(t-x)^{m_1-1}
\left( t+1\right)^{m_2-1}\nonumber \\ &\times t^{m_3-1}{\rm d}t
\nonumber
\end{align}
with $\mathrm{Re}~[m_1]>0, \mathrm{Re} ~[\mu x ]>0$.
\end{lemma}
\begin{proof}
See Appendix \ref{proof:lemma_SINR_UE_selection}.
\end{proof}
A closed-form solution for the common distribution function
$F_{\gamma(n)}(x)$ appears intractable. As we will see, closed-form
\emph{upper and lower bounds} of $F_{\gamma(n)}(x)$ can be obtained,
and these are sufficient to analyze the performance in the
large-user regime.
\begin{lemma}\label{lemma:cdf_upper_lower}
The c.d.f. of $\gamma_k(n)$ for $k \in \mathcal{U}_{n} $, $n \in \{
2, \ldots, n_T\}$ and $x \geq \frac{1}{\delta} -1$, satisfies
$F_{\tilde{\gamma}(n)}(x) \leq F_{\gamma(n)}(x) \leq
F_{\bar{\gamma}(n)}(x)$ with $F_{\tilde{\gamma}(n)}(x)$ and
$F_{\bar{\gamma}(n)}(x)$ given by
\begin{eqnarray}\label{eq:cdf_lower}
&F_{\tilde{\gamma}(n)}(x) = 1 - b_{1,n} ~ x^{ - \frac{n-1}{2}}\exp
\left(- \frac{x}{2\phi} \right) \nonumber \\
&\hspace{1cm}\times  W_{\frac{-2n_T-n+3}{2},
\frac{2n_T-n}{2}}\left(\frac{ x}{\phi} \right),
\end{eqnarray}
and
\begin{align}\label{eq:cdf_upper}
& F_{\bar{\gamma}(n)}(x) = 1 - b_{2,n}~x^{n_T-n+1}\left( x+
\frac{n_T-1}{2n_T-1}\right)^{- \frac{2n_T-n+1}{2}} \nonumber\\
& \hspace{0.5cm} \times \exp \left(- \frac{x}{2\phi} \right)
W_{\frac{-2n_T-n+3}{2}, \frac{2n_T-n}{2}}\left(\frac{ x+
\frac{n_T-1}{2n_T-1}}{\phi} \right)
\end{align}
respectively, where
\begin{eqnarray}
b_{1,n} = \frac{2^B (n-2)!}{\beta(n_T-n+1,n-1)}\left(\frac{1}{\phi}
\right)^{\frac{2n_T-n-1}{2}},\nonumber\\ b_{2,n} =\frac{2^B
(n-2)!}{\beta(n_T-n+1,n-1)}\left(\frac{1}{\phi}
\right)^{\frac{2n_T-n-1}{2}} e^{\frac{\frac{n_T-1}{2n_T-1}}{2\phi}}
,\nonumber
\end{eqnarray}
where $W_{\mu, \nu} (x)$ is Whittaker function of the second kind
\cite{Gradshteyn2000}. The equalities on both sides are approached
as $x \to \infty$.
\end{lemma}
\begin{proof}
See Appendix \ref{proof:lemma_cdf_upper_lower}.
\end{proof}

According to (\ref{eq:selection2}), $\gamma_{\mathcal{S}(n)}$ is the
maximum of a collection of i.i.d.\ random variables chosen from
$\mathcal{U}_n$, with common c.d.f.\ $F_{\gamma(n)}(x)$. Moreover,
since the analysis is carried out in the large-user regime,
according to the extreme value theory of order statistics (see e.g.\
\cite{extreme_oreder_statistics87} \cite[Appendix
I]{Mohammad08}\cite{oreder_statistics03}), the asymptotic
distribution of the largest order statistic
$\gamma_{\mathcal{S}(n)}$ depends on the \emph{tail} behavior (large
$x$) of $F_{\gamma(n)}(x)$. For $n \geq 2$, the following
closed-form asymptotic ($x \rightarrow \infty$) expansions for the
c.d.f.\ lower and upper bounds in (\ref{eq:cdf_lower}) and
(\ref{eq:cdf_upper}) are derived in Appendix
\ref{app:lemma:tail_gamma_n}:
\begin{align}\label{eq:tail_lower}
F_{\tilde{\gamma}(n)}(x) & = 1 - c_n~\phi^{n-1}
x^{-n_T-n+2}\exp\left(- \frac{x}{\phi}\right) \nonumber\\ & +
O\left( x^{-n_T-n+1} \exp\left(- \frac{x}{\phi}\right) \right),
\end{align}
and
\begin{align}\label{eq:tail_upper}
F_{\bar{\gamma}(n)}(x) & = 1- c_n ~\phi^{n-1}\frac{
x^{n_T-n+1}}{\left(x+ \frac{n_T-1}{2n_T-1}\right)^{2n_T -1}}
\exp\left(- \frac{x}{\phi}\right)\nonumber\\ & + O\left(
x^{-n_T-n+1} \exp\left(- \frac{x}{\phi}\right) \right)
\end{align}
where
\begin{equation}\label{eq:c_n}
c_n = \frac{2^B (n-2)!}{\beta(n_T-n+1,n-1)}.
\end{equation}
We can also examine the tightness of these bounds as follows. First,
noticing that both $F_{\bar{\gamma}(n)}(x)$ and
$F_{\tilde{\gamma}(n)}(x)$ are continuous monotonic increasing
functions and $\lim_{x \rightarrow + \infty} F_{\gamma(n)}(x) =
\lim_{x \rightarrow + \infty} F_{\bar{\gamma}(n)}(x) = \lim_{x
\rightarrow + \infty} F_{\tilde{\gamma}(n)}(x) =1 $, for an
arbitrary given constant $\varepsilon > 0$, we can always find a
positive number $X(\varepsilon)$ such that, whenever $x >
X(\varepsilon)$,
\begin{align}
F_{\gamma(n)}(x) - F_{\tilde{\gamma}(n)}(x) < \varepsilon/2
\nonumber
\end{align}
and
\begin{align}
F_{\bar{\gamma}(n)}(x) - F_{\gamma(n)}(x)  < \varepsilon/2.
\nonumber
\end{align}
Using (\ref{eq:tail_lower}) and (\ref{eq:tail_upper}), as $x
\rightarrow + \infty$ we have
\begin{align}
&F_{\bar{\gamma}(n)}(x) - F_{\tilde{\gamma}(n)}(x)\nonumber\\
& = c_n~\phi^{n-1} \exp\left(- \frac{x}{\phi}\right) \bigg(
x^{-n_T-n+2} \nonumber \\  & \hspace{1cm} - \frac{
x^{n_T-n+1}}{\left(x+ \frac{n_T-1}{2n_T-1}\right)^{2n_T -1}} +
o\left(  x^{-n_T-n+1}
\right)\bigg)\nonumber\\
& = c_n \left(n_T-1 \right) ~\phi^{n-1} \exp\left(-
\frac{x}{\phi}\right)\nonumber \\ & \hspace{1cm}  \times \left(
\frac{x^{n_T -n}}{\left(x+ \frac{n_T-1}{2n_T-1}\right)^{2n_T -1}} +
o\left( x^{-n_T-n+1} \right)  \right). \nonumber
\end{align}
Thus, as $x \rightarrow  + \infty$, both $F_{\bar{\gamma}(n)}(x) -
F_{\gamma(n)}(x)$ and $F_{\bar{\gamma}(n)}(x) - F_{\gamma(n)}(x)$
decrease at the speed, not slower than $O \left(\exp\left(-
\frac{x}{\phi}\right)  x^{-n_T-n+1}  \right)$.

Based on the above results, we can establish upper and lower bounds
on the asymptotic distribution of $\gamma_{\mathcal{S}(n)}$ for
large $K$ and $n=2,\ldots,n_T$. To this end, define
$\tilde{\gamma}_{\mathcal{S}(n)} = \max_{k \in \mathcal{U}_n}
\tilde{\gamma}_{k}(n)$ and $\bar{\gamma}_{\mathcal{S}(n)} = \max_{k
\in \mathcal{U}_n} \bar{\gamma}_{k}(n)$, with c.d.f.s
$F_{\tilde{\gamma}_{\mathcal{S}(n)}}(x)$ and
$F_{\bar{\gamma}_{\mathcal{S}(n)}(x)}$ respectively. It is clear
that $F_{{\tilde{\gamma}}_{\mathcal{S}(n)}}(x) \leq
F_{{\gamma}_{\mathcal{S}(n)}}(x) \leq
F_{{\bar{\gamma}}_{\mathcal{S}(n)}}(x)$ ($n=2,\ldots,n_T$). Then, we
have the following lemma:
\begin{lemma}\label{lemma:max_gamma_bound}
The random variables $\tilde{\gamma}_{\mathcal{S}(n)}$ and
$\bar{\gamma}_{\mathcal{S}(n)}$, $n \in \{ 2,\ldots, n_T \}$,
satisfy
\begin{eqnarray}\label{eq:extre_gamma_til}
&\hspace{-0.5cm}\text{Pr}\{ \chi_n - \phi \log\log \sqrt{K} \leq
\tilde{\gamma}_{\mathcal{S}(n)}
\leq  \chi_n + \phi \log\log \sqrt{K}\} \nonumber\\ & \geq 1 - O\bigg(\frac{1}{\log K}\bigg), \\
\label{eq:extre_gamma_bar} & \hspace{-0.5cm} \text{Pr}\{\chi_n -
\phi\log\log \sqrt{K} \leq \bar{\gamma}_{\mathcal{S}(n)} \leq \chi_n
+ \phi\log\log \sqrt{K}\} \nonumber\\ &   \geq 1 - O \left
(\frac{1}{\log K}\right),
\end{eqnarray}
where\footnote{Here $\log(\cdot)$ represents the natural logarithm.}
\begin{align}\label{eq:chi_n}
\chi_n = \phi \log \left( \frac{c_n K}{\phi^{n_T-1}}\right) - \phi
(n_T+n-2) \log\log \left(\frac{c_n K}{\phi^{n_T -1}} \right)
\end{align}
with $c_n$ given by (\ref{eq:c_n}).
\end{lemma}
\begin{proof}
See Appendix \ref{proof:lemma_max_gamma_bound}.
\end{proof}
The following lemma follows from the above results.
\begin{lemma}\label{lemma:sinr_up_low}
For $\gamma_{\mathcal{S}(n)}$, $n \in \{1,\ldots,n_T \}$, we have
\begin{eqnarray}\label{eq:sinr_up_low}
&\hspace{-0.5cm}\text{Pr}\{ \chi_n - \phi \log\log \sqrt{K} \leq
\gamma_{\mathcal{S}(n)} \leq \chi_n + \phi \log\log
\sqrt{K}\}\nonumber \\& \geq 1 - O\bigg(\frac{1}{\log K}\bigg),
\end{eqnarray}
where
\begin{equation}\label{eq:chi_1}
\chi_1 = \phi \log \left( \frac{2^B K}{\phi^{n_T-1}}\right) -
\phi(n_T-1) \log\log \left( \frac{2^B K}{\phi^{n_T-1}}\right)
\end{equation}
and $\chi_n$ for $n=2,\ldots,n_T$ is given by (\ref{eq:chi_n}).
\end{lemma}
\begin{proof}
See Appendix \ref{proof:lemma_sinr_up_low}.
\end{proof}

We can now prove the following theorem (see Appendix
\ref{proof:theorem_sum_rate_THP_feedback}), which presents a key
contribution:
\begin{theorem}\label{theorem:sum_rate_THP_feedback}
For a fixed number of transmit antennas $n_T$ and fixed transmit
power $P$, the sum rate $R_{\text{G-THP-Q}}$ of the proposed
quantized CSI-based TH precoding scheme satisfies
\begin{eqnarray}\label{eq:converge_1}
\lim_{K \rightarrow \infty}
\frac{R_{\text{G-THP-Q}}}{n_T\log_2[\varrho\log K]} = 1
\end{eqnarray}
with probability 1, where $\varrho = \frac{P}{n_T}$. For $M$ fixed,
\begin{eqnarray}\label{eq:converge_2}
& \lim_{K \rightarrow  \infty } \mathbb{E} \{R_{\text{BC}}\} -
\mathbb{E} \{ R_{\text{G-THP-Q}}\}  \nonumber \\ & \hspace{1cm}\leq
n_T \log_2\left(1+ \frac{1}{M-1}\right).
\end{eqnarray}
Moreover, as $M, K\rightarrow \infty$ we have
\begin{eqnarray}\label{eq:converge_4}
& \hspace{-4cm}\mathbb{E} \{ R_{\text{BC}}\} - \mathbb{E} \{
R_{\text{G-THP-Q}} \} \nonumber\\ & \leq \min\left\{O \left
(\frac{\log \log K}{\log K} \right ), O\left(\frac{1}{M-1}
\right)\right\},
\end{eqnarray}
and
\begin{eqnarray}\label{eq:converge_3}
\lim_{K \rightarrow  \infty,~  M \rightarrow  \infty } \mathbb{E} \{
R_{\text{BC}}\} - \mathbb{E} \{ R_{\text{G-THP-Q}} \} =0,
\end{eqnarray}
where $R_{\text{BC}}$ denotes the sum rate of the MIMO broadcast
channel, achieved with DPC.

\end{theorem}

\subsection{High SNR or Interference-Limited Regime}
In this subsection, we let SNR go large. In this regime, for each
user $k \in \mathcal{U}_1$, the output SINR becomes
\begin{equation}
\lim_{P\rightarrow \infty} \gamma_{k}(n) = \frac{ \cos^2\theta_k}{
\sin^2\theta_k} \defeq \hat{\gamma}_{k}(1),
\end{equation}
whose c.d.f. has been obtained in \cite{Yoo_limited} as
\begin{align}\label{eq:cdf1_highSNR}
F_{\hat{\gamma}(1)}(x) =  \left\{ \begin{array}{ll} 1-
\frac{2^B}{(1+x)^{n_T-1}}, &
x\geq \frac{1}{\delta}-1 \\
0, &\text{otherwise}
\end{array}
\right..
\end{align}
For each user $k \in \mathcal{U}_n, n=2,\ldots, n_T$, the output
SINR becomes
\begin{equation}
\lim_{P\rightarrow \infty} \gamma_{k}(n) = \omega_k(n) \frac{
\cos^2\theta_k}{ \sin^2\theta_k} \defeq \hat{\gamma}_{k}(n) ,
\end{equation}
whose c.d.f. is given by the following lemma.
\begin{lemma}\label{lemma:cdfn_highSNR}
For each user $k \in \mathcal{U}_n, n =2, \ldots, n_T$ and $x \geq
\frac{1}{\delta} -1 $, the c.d.f. of $\hat{\gamma}_{k}(n)$ is given
by
\begin{align}
& F_{\hat{\gamma}(n)} (x) = 1 - \frac{d_n}{x^{n_T-1}} \nonumber\\&
 ~\times{}_2F_1\left(n_T-1,2n_T-n;2n_T-1; -\frac{1}{x} \right),
\end{align}
where ${}_2F_1 (\cdot)$ is a generalized hypergeometric function
\cite[9.14]{Gradshteyn2000} and $d_n = \frac{2^B \beta(2n_T -n,
n-1)}{\beta(n_T-n+1,n-1)} $.
\end{lemma}
\begin{proof}
See Appendix \ref{proof:lemma_cdfn_highSNR}.
\end{proof}
The tail distribution of $F_{\hat{\gamma}(n)}(x)$ ($x \to \infty$)
is obtained as
\begin{align}\label{eq:tail_cdfn_highSNR}
F_{\hat{\gamma}(n)}(x) =  1- \frac{d_n}{x^{n_T-1}} +
O\left(\frac{1}{x^{n_T}}\right) ,
\end{align}
where we have used the series expansion ${}_2F_1(a,b;c;x) =
\sum_{l=0}^{\infty} \frac{(a)_l (b)_{l}}{(c)_{l}} \frac{x^{l}}{l!}$,
and $(\cdot )_r $ is the Pochammer symbol, which is defined as
$(a)_k = a(a + 1)\cdots (a + k - 1)$ with $(a)_0 = 1$.

The output SINR of the selected user at the $n$-th iteration becomes
$\hat{\gamma}_{\mathcal{S}(n)} = \arg \max_{k \in \mathcal{U}_n}
\hat{\gamma}_k (n) $, whose extremal distribution is given by the
following lemma.
\begin{lemma}\label{lemma:extream_cdfn_highSNR}
For $\hat{\gamma}_{\mathcal{S}(1)}$, we have
\begin{align}\label{eq:extre_cdf1_highSNR}
&\text{Pr} \bigg\{ \left(\frac{2^B
K}{\log\sqrt{K}}\right)^{\frac{1}{n_T-1}}  -1 \leq
\hat{\gamma}_{\mathcal{S}(1)} \leq \nonumber \\&
\hspace{0.5cm}\left(2^B K\log\sqrt{K}\right)^{\frac{1}{n_T-1}}  -1
\bigg\}  \geq 1- O \left(\frac{1}{\log K}\right),
\end{align}
and for $\hat{\gamma}_{\mathcal{S}(n)}$, $n =2,\ldots,n_T$, we have
\begin{align}\label{eq:extre_cdfn_highSNR}
&\text{Pr} \left\{ \left(\frac{d_n
K}{\log\sqrt{K}}\right)^{\frac{1}{n_T-1}} \leq
\hat{\gamma}_{\mathcal{S}(n)} \leq \left(d_n
K\log\sqrt{K}\right)^{\frac{1}{n_T-1}} \right \} \nonumber \\ &
\hspace{1cm}\geq 1- O \left(\frac{1}{\log K}\right),
\end{align}
where $d_n$ is defined in \emph{Lemma} \ref{lemma:cdfn_highSNR}.
\end{lemma}
\begin{proof}
See Appendix \ref{proof:lemma_extream_cdfn_highSNR}
\end{proof}
With \emph{Lemma} \ref{lemma:extream_cdfn_highSNR}, for large $K$
the average sum rate can be approximated as
\begin{align}\label{eq:sum_rate_highSNR}
& \hspace{-0.5cm}\mathbb{E}\left\{ R_{\text{high-SNR}}\right\} \nonumber\\
&= \mathbb{E}\left\{
\sum_{i=1}^{n_T}\log_2\left(1+\hat{\gamma}_{\mathcal{S}(i)} \right)
\right\} \nonumber\\
&\thickapprox \frac{n_T}{n_T-1} \left(B + \log_2 K \right)  +
\frac{1}{n_T-1} \nonumber\\& \hspace{0.5cm} \times \sum_{n=2}^{n_T}
\log_2 \left(\frac{\beta(2n_T -n, n-1)}{\beta(n_T-n+1,n-1)} \right)
\nonumber\\& \hspace{0.5cm} + O(\log_2\log K),
\end{align}
where $R_{\text{high-SNR}}$ represents $R_{\text{G-THP-Q}}$ in
(\ref{eq:throughtput}) in the high-SNR regime ($P \rightarrow \infty
$).

\subsection{Discussion of Results}\label{subsec:discusses}
In this subsection, we will discuss the analytical results obtained
above.
\begin{enumerate}

\item
Asymptotically, (41) shows that for fixed finite $n_T$ and SNR, our
scheme can achieve the maximum spatial multiplexing gain of $n_T$,
and also the maximum multi-user diversity gain up to first order
(i.e., the SNR scales with $\log K$, and the average sum rate scales
as $\log \log K$).

\item
As explained in Section \ref{sec:model}, the precoding
loss\cite{Windpassinger04} caused by TH precoding can be negligible
for moderate sizes $M$ and vanishes as $M$ increases. \emph{Theorem}
\ref{theorem:sum_rate_THP_feedback} shows that if we ignore this
loss, the asymptotic sum rate performance achieved by this quantized
CSI-based TH precoding can converge to that of the MIMO broadcast
channel. 

\item
The asymptotic results in (\ref{eq:chi_n})--(\ref{eq:chi_1}) are
valid only when both $K$ and $\frac{2^B K }{\phi^{n_T-1}}$ are
large. Moreover, the c.d.f.s (\ref{eq:cdf_exact1}) and
(\ref{eq:cdf_exact2}) are valid only for $\gamma_{\mathcal{S}(n)}
\geq 2^{\frac{B}{n_T-1}}-1$. Thus, when either $B$ or $\phi (P)$ is
sufficiently large that a given $K$ is not large enough to satisfy
the aforementioned conditions, then some of these conditions may
fail. Also given a target SINR or sum rate performance, the $B$, $K$
and $\phi (P)$ should scale such that $\frac{c_n K}{\phi^{n_T -1}} $
is a constant. Or equivalently,
\begin{align} B +\log_2 K = (n_T -1) \log_2 P +c \end{align}
for some constant $c$.  In \cite{Yoo_limited}, a very similar result
has also been observed for the multi-user diversity gain in the
first-order terms $O(\log K)$ for the downlink MU-MIMO systems,
employing ZFBF based on the quantized CSI.

\item
It is instructive to compare the performance of the proposed
algorithm with quantized CSI with that of the ZFDPC algorithm with
perfect CSI. For the system with perfect CSI, we refer to our
previous work in \cite{L.sun_multiuser}. We have noted above that TH
precoding is a practical implementation of ZFDPC proposed in
\cite{Caire03}. If we ignore the precoding loss and shaping loss,
the performance of TH precoding is the same as that of ZFDPC.
Moreover, although the additional semi-orthogonal constraint is
imposed by the ZFDPC algorithm in \cite{L.sun_multiuser}, it was
shown that this constraint asymptotically does not reduce the
multi-user diversity gain in either the first ($\log(K)$) or the
second-order ($\log\log(K)$) terms. For the sake of clarity, in the
following we rewrite the asymptotic distribution of the output SNR
$\zeta_{\pi(n)}$ for the $n$-th selected user which was derived in
\cite{L.sun_multiuser}.
\begin{eqnarray}\label{eq:sinr_up_low2}
& \hspace{-2cm}\text{Pr}\{ \varpi_n - \varrho \log\log \sqrt K \leq
\zeta_{\pi(n)} \leq  \nonumber\\& \upsilon_n + \varrho \log\log
\sqrt K  \} \geq 1 - O\bigg( \frac{1}{\log K} \bigg),
\end{eqnarray}
where $\varpi_n = \varrho \log \bigg(\frac{K}{\varepsilon_n}\bigg) +
\varrho (n_T-n) \log\log \bigg(\frac{K }{\varepsilon_n}\bigg)$, $
\upsilon_n = \varrho \log \left(\frac{K}{\epsilon_n}\right) +
\varrho (n_T-n) \log\log\left(\frac{K}{\epsilon_n}\right)$ and
$\varepsilon_n = \frac{ \Gamma(n_T-n+1)(n-1)^{n-1}}{\Gamma(n)}$,
$\epsilon_n = \Gamma(n_T-n+1)$.
If we compare the asymptotic distribution of the output SINR for the
quantized CSI case with that of the output SNR for the perfect CSI
case, which are given in (\ref{eq:sinr_up_low}) and
(\ref{eq:sinr_up_low2}) respectively, we can see that there is an
additional term $\Delta = \frac{2^B}{\phi^{n_T-1}}$ that affects the
multi-user diversity gain in both first-order terms $O(\log K)$ and
second-order terms $O(\log\log K)$ (omitting the constant terms only
related to $n_T$ and $n$). We can see that, for a given SNR,
increasing the feedback rate can improve the effect of multiuser
diversity gain and increase the sum rate. For a fixed feedback rate,
increasing transmit power can reduce the effect of the multiuser
diversity gain.

\item
The asymptotic distributions in (\ref{eq:sinr_up_low}) and
(\ref{eq:sinr_up_low2}) also show a very interesting result that,
for the perfect CSI case, whilst the coefficient of the first order
term $O(\log K)$ is unaffected by the number of iterations $n$, the
coefficient of the second-order term $O(\log \log K)$ increases
linearly with $n_T$ and decreases linearly with $n$. In contrast,
for the quantized CSI case, the coefficient of the second-order term
$O(\log \log K)$ decreases linearly with both $n_T$ and $n$. This
difference between perfect CSI and quantized CSI cases in the
coefficient of the term $O(\log \log K)$ is also caused by MUI given
by the second term of (\ref{eq:Rx_sig_limited}).

\item

Although the quantized CSI-based and the perfect CSI-based TH
precoding schemes achieve the same asymptotic average sum rate as $K
\rightarrow \infty$, the speed of convergence to this optimal sum
rate can be very different. 
This performance difference is due to the MUI caused by quantized
CSI at the transmitter side. We can see this clearly from the last
line of (\ref{eq:difference}) in the proof of \emph{Theorem}
\ref{theorem:sum_rate_THP_feedback}. In fact, the term $O
\left(\frac{1}{\log K}\right)$ in the first term of
(\ref{eq:difference}) contains $\Delta$ which is due to MUI. For
finite $K$, this term can been even larger than the term $O
\bigg(\frac{\log\log K}{\log K} \bigg)$. Thus, for finite $K$ there
is a gap in the average sum rates between the two cases.

\item
In the high-SNR regime, we can also observe the similar
interchangeability between $B$ and $\log K$ as we have observed for
general SNR values. In addition, from (\ref{eq:sum_rate_highSNR}) it
is easy to see that for $K$ and finite $B$ the sum rate converges to
a constant as $P \rightarrow \infty$. In contrast to the previous
finding for the fixed finite SNR regime that the asymptotic sum rate
scaling increases only via a factor of $\log_2 \log K$ in the
multi-user diversity gain, the asymptotic sum rate
(\ref{eq:sum_rate_highSNR}) grows via a factor of $\log_2 K $. Thus,
multiuser diversity gain is even more beneficial in high-SNR regime.
However, the spatial multiplexing gain decreases remarkably from
$n_T$ to $\frac{n_T}{n_T-1}$. The same results can also be observed
in \cite{Yoo_limited} in the context of ZFBF.

\end{enumerate}

\section{Numerical Results}
In this section we present some numerical results. We assume $n_T =
4$ and $M$ is large enough such that the precoding loss can be
ignored\footnote{The modulation order $M$ only affects the sum rate
performance through the precoding loss $\frac{M}{M-1}$. By using
sufficiently large $M$, the precoding loss can arbitrarily approach
$1$ (no precoding loss). Thus, in simulations we have ignored the
effect of modulation order.}. The SNR of the system is defined to be
equal to $P$.

In Fig. \ref{fig:sum_rate_SNR}, we present the average sum rates
versus SNR $P$ for a system with $K = 100$ users, and various CDI
quantization levels $B = 8, 12$ bits. As expected, it is seen from
the figure that, with quantized CSI, the average sum rates of TH
precoding and ZFBF both approach the corresponding average sum rates
with perfect CSI as $B$ increases. We can also observe that, with
perfect CSI, when the number of users is large, the average sum
rates of TH precoding and ZFBF can be quite similar. However, the
average sum rates of TH precoding and ZFBF with quantized CSI can be
quite different.

\begin{figure}[t]
\centering
\includegraphics[width= 0.95\columnwidth]{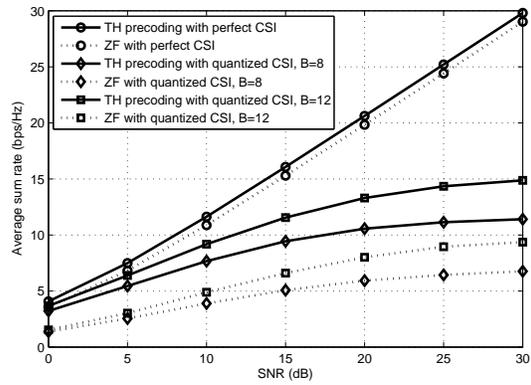}
\caption{The average sum rate versus system SNR $P$ under $n_T =4$
and various $B$.} \label{fig:sum_rate_SNR}
\end{figure}
\begin{figure}[t]
\centering
\includegraphics[width= 0.95\columnwidth]{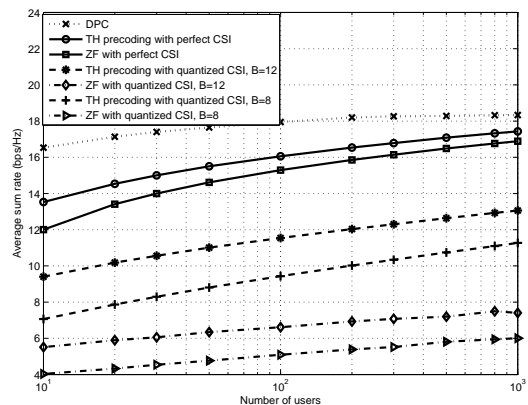}
\caption{The average sum rate versus number of users $K$ under $n_T
=4$, $P= 15$ dB and various $B$.} \label{fig:sum_rate_user}
\end{figure}
\begin{figure}[t]
\centering
\includegraphics[width= 0.95\columnwidth]{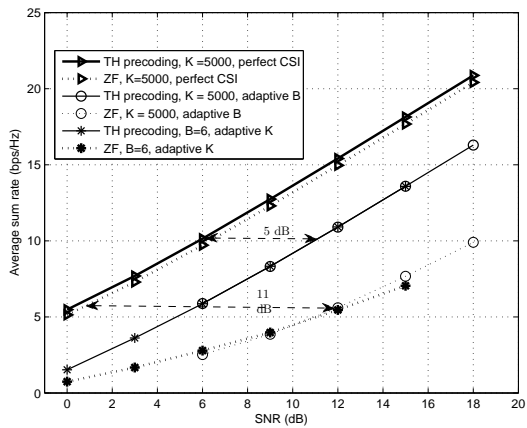}
\caption{The average sum rate versus SNR $P$ under $n_T =4$, and
adaptive $B$ or $K$ such that $B + \log_2 K = (n_T-1) \log_2 P +
8.32$.} \label{fig:sum_rate_relation}
\end{figure}
\begin{figure}[t]
\centering
\includegraphics[width= 0.95\columnwidth]{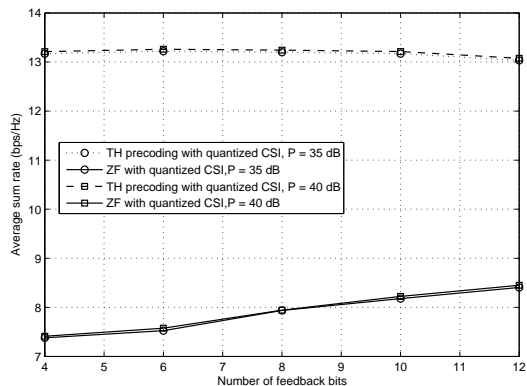}
\caption{The average sum rate versus the number of feedback bits $B$
in high-SNR regime ($P=35$ dB, $P=40$ dB) under $n_T =4$, and
adaptive $B$ or $K$ such that $B + \log_2 K = 16.55$.}
\label{fig:sum_rate_highSNR_relation}
\end{figure}

Fig. \ref{fig:sum_rate_user} plots the average sum rate as a
function of the number of users achieved by TH precoding and ZFBF
based on both perfect CSI and quantized CSI at the transmitter side.
A curve is also presented for the optimal DPC which acts as an
achievable upper bound, and is computed using the algorithm from
\cite{Vishwanath03}. As evident from the figure, the average sum
rate of TH precoding scheme with perfect CSI converges slowly to the
average sum rate achieved with optimal DPC as $K$ grows large. The
curves for the quantized CSI cases are far below those for the
perfect CSI cases at finite $K$, which is due to MUI as explained in
Section \ref{sec:asymptotic_sum_rate_perforamnce}  for TH precoding
and in \cite{Yoo_limited} for ZFBF, respectively. Still, we can see
the trend of convergence to the optimal DPC for the quantized CSI
case. In addition, TH precoding performs significantly better than
linear ZFBF for both the perfect CSI and quantized CSI cases. For
finite $K$, there is a gap between the average sum rates achieved by
TH precoding and ZFBF for both the perfect CSI and quantized CSI
cases. Moreover, the gap between the average sum rates achieved by
TH precoding and ZFBF for the systems with quantized CSI increases
as $K$ increases, whereas this gap in the systems with perfect CSI
decreases as $K$ increases. In fact, we see that the curve for ZFBF
with perfect CSI converges slowly to that of DPC (thus, converge to
that of TH precoding) as $K$ grows large. However, the speed of
convergence is slower than that of TH precoding. This result has
also been proved in \cite{L.sun_multiuser}.

Fig. \ref{fig:sum_rate_relation} investigates the interchangeability
between $B$ and $\log_2 K$, where we adapt $B$ and $K$ as
\begin{align}\label{eq:relation_BK}
B -B_0 + \log_2\left(\frac{K}{K_0} \right)  = (n_T-1) \log_2 \left(
\frac{P}{P_0} \right),
\end{align}
so that a constant gap from the sum rate of the perfect CSI case can
be maintained. We consider a system in the large-user regime such
that $K_0 = 5000, B_0=6$, $P_0 = 10 $ dB. This corresponds to $B +
\log_2 K = (n_T -1)\log_2 P + 8.32$. We either fix $K = K_0$ and
adapt $B$ as a function of $P$, or fix $B = B_0$ and adapt $K$ as a
function of $P$. For comparison, we plot the results for both TH
precoding and ZFBF. We find that, in the large-user regime, the
average sum rate achieved by TH precoding and ZFBF with perfect CSI
are very close to each other, whereas the corresponding results with
quantized CSI are very different. Particularly, we can observe a
constant SNR gap of about $5$ dB from the perfect CSI curve for TH
precoding, which confirms the analysis in Section
\ref{sec:asymptotic_sum_rate_perforamnce}. Whereas the corresponding
SNR gap for ZFBF is about $11$ dB from the perfect CSI curve and
this gap increases a little as SNR increases in the large-user
regime.

Fig. \ref{fig:sum_rate_highSNR_relation} investigates the
interchangeability between $B$ and $\log_2 K$ in the high-SNR regime
($P=35$ dB, $P=40$ dB) for both TH precoding and ZFBF, where we
adapt $B$ and $K$ as (\ref{eq:relation_BK}).
We consider a system in the large-user high-SNR regime such that
$K_0 = 3000, B_0=5$ and $P_0 = P$. This corresponds to $B + \log_2 K
= 16.55$. We adapt $K$ as a function of $B$. For comparison, we plot
the results for both TH precoding and ZFBF. It is seen that the
average sum rate achieved by TH precoding almost remains constant as
$B$ increases, and also doesn't change much as SNR increases from
$35$ dB to $40$ dB, which aligns with the
analysis in Section IV. 
However, the average sum rate of ZFBF increases a little as the
number of feedback bits $B$ increases.
\begin{figure*}[!t]
\begin{eqnarray}\label{eq:joint_density_2}
f (t_1,\ldots,t_{n-1}) = \left\{ \begin{array}{ll}
\frac{\Gamma(n_T)}{\Gamma(n_T-n+1)} \left(1 - \sum_{i=1}^{n-1}t_i
\right)^{n_T-n}, & t_i \geq 0,  i= 1,\ldots, n-1,
\sum_{i=1}^{n-1} t_i \leq 1\\
0, &\text{otherwise}
\end{array}
\right..
\end{eqnarray}
\hrulefill
\end{figure*}

\section{Conclusion}
We have investigated the sum rate of a quantized CSI-based TH
precoding scheme for a MU-MIMO system, denoted G-THP-Q. We have
proved that for fixed finite SNR and finite $n_T$, G-THP-Q can
achieve the optimal sum rate scaling of the MIMO broadcast channel,
and that the average sum rate converges to the average sum capacity
of the MIMO broadcast channel as the number of users $K$ grows
large. In addition, our analytical results have provided important
insights into the effect of quantized CSI on the multiuser diversity
gain in both the first- and the second-order terms. We have studied
and derived the tradeoffs between the number of feedback bits, the
number of users and SNR. Particularly, a constant SNR gap from the
result of the perfect CSI case can be achieved by simultaneously
interchanging the number of users and feedback bits.

\appendices

\section{Proof of Lemma
\ref{lemma:pdf_omega_n}}\label{proof:lemma_pdf_omega_n} According to
\emph{Lemma} \ref{lemma:iid_property}, for large $K$, the channel
direction vector $\bar{\mathbf{h}}_k$, for $k\in \mathcal{U}_n$, is
an isotropically distributed unit vector on the $n_T$-dimensional
complex unit hypersphere. Using RVQ, the same results holds for the
quantized channel direction vector $\hat{\mathbf{h}}_k$, for $k\in
\mathcal{U}_n$. In addition, the subspace spanned by the orthonormal
basis $\hat{\mathbf{q}}_1, \ldots, \hat{\mathbf{q}}_{n-1}$ becomes
independent of $\bar{\mathbf{h}}_k$, for $k\in \mathcal{U}_n$. Thus,
without loss of generality we can assume $\mathbf{q}_i=
\mathbf{e}_i$, where $\mathbf{e}_i$ is the $i$-th row of the
identity matrix $\mathbf{I}_{n_T}$.

For $k\in \mathcal{U}_n$, let $\hat{\mathbf{h}}_k = [m_1, \cdots,
m_{n_T}]^T$, and $t_i = | \hat{\mathbf{h}}_k \hat{\mathbf{q}}_i^H
|^2 $. Then, with $\mathbf{q}_i= \mathbf{e}_i, i=1,\ldots,n-1$, we
have $t_i = | m_i|^2 $ for $1 \leq i \leq n-1$. The joint p.d.f. of
$t_1,\ldots,t_{n-1}$ has been obtained in \cite{L.sun_multiuser} and
is shown in (\ref{eq:joint_density_2}) at the top of this page. Now
we want to obtain the distribution of $w_1= \sum_{l=1}^{n-1} t_l$.
We define the following transformation of variables
\begin{align}\label{eq:change_variable}
w_1 =\sum_{l=1}^{n-1} t_l, ~~w_i=t_i ~~~\text{for}~ i =2, \ldots,
n-1.
\end{align}
The corresponding Jacobian can be easily obtained as $J = 1$. Thus
the joint p.d.f. of $w_1, \ldots, w_{n-1}$ is
\begin{eqnarray}
& \hspace{-2cm }f_{w_1,\ldots,w_{n-1}}\left( x_1, \ldots,
x_{n-1}\right) \nonumber
\\ & = \frac{\Gamma(n_T)}{\Gamma(n_T-n+1)} \left(1 - x_1
\right)^{n_T-n}.
\end{eqnarray}
Since $ 0 \leq t_i \leq 1$, we have $0 \leq t_1 = w_1 -
\sum_{i=2}^{n-1}w_i \leq 1$. According to
(\ref{eq:change_variable}), the region of the random variables after
transformation can be obtained as $\mathcal{D} = \big\{ \left( w_1,
\ldots, w_{n-1} \right) ~|~ 0\leq \sum_{i=2}^{n-1} w_i\leq w_1 \leq
1,  0 \leq w_i \leq 1 ~\text{for} ~i = 2, \ldots, n-1 \big\}$. Then
the marginal distribution of $w_1$ can be obtained as
\begin{eqnarray}
&& \hspace{-0.5cm}f_{w_1} (x) \nonumber \\ && = \int \cdots
\int_{\mathcal{D}} f
(x,x_2 \ldots,x_{n-1}) ~{\rm d}x_2 \cdots{\rm d}x_{n-1} \nonumber \\
&& = \int \cdots \int_{\mathcal{D}}
\frac{\Gamma(n_T)}{\Gamma(n_T-n)} \left(1 - x_1 \right)^{n_T-n}
~{\rm d}x_2 \cdots{\rm d}x_{n-1} \nonumber\\
&& \mathop = \limits^{(a)} \frac{\Gamma(n_T)}{\Gamma(n_T-n)}  (1-
x)^{n_T-n} \frac{x^{n-2}}{\left(n-2\right) !} \nonumber,
\end{eqnarray}
where in $(a)$ we have used the identity $\mathop {\int \int {
\cdots \int {}} }\limits_{\scriptstyle \sum_{i=1}^{n}t_i \leq h
\hfill \atop \scriptstyle t_1 \geq 0, \ldots, t_{n} \geq 0  \hfill}
~{\rm d} t_1 \cdots ~{\rm d} t_n = \frac{h^n}{n!} $
\cite{Gradshteyn2000}. We find that $w_1$ follows beta distribution
with shape parameters $(n-1)$ and $ (n_T-n+1)$.
According to (\ref{eq:sum_l}) we have $\omega_k(n) =
1-\sum_{j}^{n-1}| \hat{\mathbf{h}}_{k} \hat{\mathbf{q}}_{j}^H |^2 =
1- w_1$. Thus $\omega_k(n) \sim \mathrm{Beta}(n_T-n+1,n-1)$, whose
p.d.f. is given by (\ref{eq:pdf_omega_n}).

\section{Proof of \emph{Lemma}
\ref{lemma:SINR_UE_selection}}\label{proof:lemma_SINR_UE_selection}
From \emph{Lemma} \ref{lemma:S_I_distribution} it is clear that
$\gamma_{k}(n)$ has the same distribution of $\gamma = \frac{\phi
\omega \left( X + (1-\delta)Y \right)}{\phi \delta Y +1}$, where
$\omega$ has the same distribution of $\omega_k(n)$, and $X, Y$ are
defined in \emph{Lemma} \ref{lemma:S_I_distribution}. Then we will
derive the c.d.f. of $\gamma$ as
\begin{align}
& \mathrm{Pr}(\gamma \geq x) \nonumber \\ &= \mathrm{Pr} \left(
\frac{\phi \omega \left( X +
(1-\delta)Y \right)}{\phi \delta Y +1} \geq x\right)\nonumber\\
& = \int_{0}^{1} f_{\omega}(z){\rm d}z  \int_{0}^{\infty}
\mathrm{Pr}\left( X \geq \left(\frac{\delta}{z}y+ \frac{1}{\phi
z}\right)x - (1-\delta)y \right) \nonumber \\ & \hspace{1cm} \times
f_{Y}(y) {\rm d}y.
\end{align}
If $x \geq \frac{1-\delta}{\delta}$, then $\left(\frac{\delta}{z}y+
\frac{1}{\phi z}\right)x - (1-\delta)y$ is nonnegative for any $y
\geq 0$ and $0 \leq z \leq 1$. Using the definition of random
variables $X$ and $Y$, the above can be further written as
\begin{align}
& \mathrm{Pr}(\gamma \geq x) \nonumber \\ &= \int_{0}^{1}
f_{\omega(n)}(z){\rm d}z \int_{0}^{\infty}  e^{-
\left(\frac{\delta}{z}y+ \frac{1}{\phi z}\right)x +
(1-\delta)y}f_{Y}(y) {\rm d}y \nonumber\\
& = \int_{0}^{1} e^{-\frac{x}{\phi z}} f_{\omega}(z){\rm d}z
\int_{0}^{\infty}
e^{- \left[\frac{\delta}{z}x -(1-\delta)\right] y } \frac{y^{n_T-2}e^{-y}}{(n_T-2)!} {\rm d}y  \nonumber\\
&=  \int_{0}^{1} e^{-\frac{x}{\phi z}} f_{\omega}(z){\rm d}z
\int_{0}^{\infty} \frac{ e^{-  \delta \left( \frac{x}{z}+1 \right) y
}y^{n_T-2}}{(n_T-2)!} {\rm d}y  \nonumber\\
&= \int_{0}^{1}\frac{2^B e^{-\frac{x}{\phi z}}}{ \left(
\frac{x}{z}+1 \right)^{n_T-1}}
\frac{1}{\beta(n_T-n+1,n-1)} \nonumber \\ & \hspace{1cm} \times z^{n_T-n}(1- z)^{n-2}{\rm d}z\nonumber\\
\label{eq:int_1} &= c_n \int_{0}^{1}\frac{ e^{-\frac{x}{\phi z}}}{
\left( z+ x \right)^{n_T-1}}z^{2n_T-n-1}(1- z)^{n-2}{\rm d}z,
\end{align}
where $c_n$ is as defined in the lemma statement.

With the change of variables $t = \frac{x}{z}$, (\ref{eq:int_1}) can
be written as
\begin{align}
& \mathrm{Pr}(\gamma \geq x) \nonumber \\ &=  c_n
\int_{x}^{\infty}\frac{ e^{-\frac{t}{\phi }}}{ \left( \frac{x}{t}+ x
\right)^{n_T-1}}\left( \frac{x}{t} \right)^{2n_T-n-1}  \left(1-
\left( \frac{x}{t} \right)\right)^{n-2}\nonumber
\\ &\hspace{2cm}\times\frac{x}{t^2}{\rm d}t\nonumber\\
&=c_n  x^{n_T-n+1} \int_{x}^{\infty} e^{-\frac{t}{\phi }}
\frac{(t-x)^{n-2}}{\left( t+1\right)^{n_T-1}t^{n_T}}{\rm
d}t\nonumber\\
&=c_n  x^{n_T-n+1}  V(n-1;-n_T+2;-n_T+1;\frac{1}{\phi}; x).
\end{align}
Thus \emph{Lemma} \ref{lemma:SINR_UE_selection} is proved.

\section{Proof of \emph{Lemma}
\ref{lemma:cdf_upper_lower}}\label{proof:lemma_cdf_upper_lower}
Using the geometric-mean inequality we have
\begin{align}
(t+1)^{n_T-1}t^{n_T} & < \left( \frac{(n_T-1)(t+1)+{n_T}t}{2n_T-1}
\right)^{2n_T -1} \nonumber\\
& = \left( t+ \frac{n_T-1}{2n_T-1} \right) ^{2n_T-1}.
\end{align}
Thus, for $ \frac{1}{\delta}< x < \infty$,
\begin{align}
& V(n-1;-n_T+2;-n_T+1;\frac{1}{\phi}; x) \nonumber\\ &
> \int_{x}^{\infty} e^{-\frac{t}{\phi }} \frac{(t-x)^{n-2}}{\left(
t+
\frac{n_T-1}{2n_T-1} \right) ^{2n_T-1}}{\rm d}t \nonumber\\
\label{eq:V_1} &=  \int_{x+ \frac{n_T-1}{2n_T-1} }^{\infty}
e^{-\frac{s-\frac{n_T-1}{2n_T-1}  }{\phi }} \frac{\left(s-x-
\frac{n_T-1}{2n_T-1} \right)^{n-2}}{s^{2n_T-1}}~ {\rm d}s\\
\label{eq:V_2} & = e^{ \frac{\frac{n_T-1}{2n_T-1}}{\phi }}
\left(\frac{1}{\phi} \right)^{\frac{2n_T-n-1}{2}}\Gamma(n-1) \nonumber\\
&  \hspace{0.5cm}\times  \left( x+ \frac{n_T-1}{2n_T-1}\right)^{
\frac{n-2n_T-1}{2}} \exp \left(- \frac{x+ \frac{n_T-1}{2n_T-1}
}{2\phi} \right) \nonumber\\
&  \hspace{0.5cm} \times W_{\frac{-2n_T-n+3}{2}, \frac{2n_T-n}{2}}
\left(\frac{ x+ \frac{n_T-1}{2n_T-1}}{\phi}
\right)\\
& =  \Gamma(n-1)\left(\frac{1}{\phi} \right)^{\frac{2n_T-n-1}{2}}
e^{ \frac{\frac{n_T-1}{2n_T-1}}{2 \phi }}\nonumber\\
& \hspace{0.5cm} \times \left( x+ \frac{n_T-1}{2n_T-1}\right)^{ -
\frac{2 n_T-n+1}{2}}  \exp \left(- \frac{x}{2\phi} \right) \nonumber\\
& \hspace{1cm} \times W_{\frac{-2n_T-n+3}{2},
\frac{2n_T-n}{2}}\left(\frac{ x+ \frac{n_T-1}{2n_T-1}}{\phi} \right)
,
\end{align}
where the change of variables $s=  t + \frac{n_T-1}{2n_T-1}$ is
employed in (\ref{eq:V_1}), whilst (\ref{eq:V_2}) is obtained by
using \cite[3.383.4]{Gradshteyn2000}. In addition, we have
\begin{align}
& V(n-1;-n_T+2;-n_T+1;\frac{1}{\phi}; x) \nonumber\\
& < \int_{x}^{\infty} e^{-\frac{t}{\phi }}
\frac{(t-x)^{n-2}}{t^{2n_T-1}}{\rm d}t\nonumber\\
\label{eq:V_3} &= \Gamma(n-1)\left(\frac{1}{\phi}
\right)^{\frac{2n_T-n-1}{2}}x^{ - \frac{2 n_T-n+1}{2}}\exp \left(-
\frac{x}{2\phi} \right) \nonumber\\
& \hspace{1cm} \times W_{\frac{-2n_T-n+3}{2}, \frac{2
n_T-n}{2}}\left(\frac{ x}{\phi} \right),
\end{align}
where we have used \cite[3.383.4]{Gradshteyn2000} again to obtain
(\ref{eq:V_3}). Then, for $ \frac{1}{\delta} < x < \infty$ and
$n\geq 2$, $F_{\bar{\gamma}(n)}(x) $ and $F_{\tilde{\gamma}(n)}(x) $
are obtained by individually substituting (\ref{eq:V_2}) and
(\ref{eq:V_3}) into (\ref{eq:cdf_exact2}) with $\Gamma(n-1)=
(n-2)!$. In addition, it is easy to prove that
$F_{\tilde{\gamma}(n)}(0) =0$, $F_{\tilde{\gamma}(n)}(1) = 1$ and
$F_{\tilde{\gamma}(n)}(x)$ is an increasing function of $x$. Thus,
$F_{\tilde{\gamma}(n)}(x)$ is a distribution function. For the same
reasons, $F_{\bar{\gamma}(n)}(x)$ is also a distribution function.
\emph{Lemma} \ref{lemma:cdf_upper_lower} is proved.

\section{Asymptotic expansion of c.d.f.s of $\tilde{\gamma}_k(n)$ and $\bar{\gamma}_k(n)$
for large $x$.}\label{app:lemma:tail_gamma_n} Using the asymptotic
expansion of the Whittaker function $W_{\lambda, \mu}(x)$ for large
$x$ given by \cite[9.227]{Gradshteyn2000}, we have
\begin{align}\label{eq:expansion_1}
& W_{\frac{-2n_T-n+3}{2}, \frac{2n_T-n}{2}}\left(\frac{ x+
\frac{n_T-1}{2n_T-1}}{\phi} \right) \nonumber\\& = \exp\left(-\frac{
x+ \frac{n_T-1}{2n_T-1}}{2 \phi}\right)\left(\frac{ x+
\frac{n_T-1}{2n_T-1}}{\phi}\right)^{\frac{-2n_T-n+3}{2}} \nonumber\\
& \hspace{1cm} \times \left( 1+ O \left(\frac{1}{ x+
\frac{n_T-1}{2n_T-1}} \right) \right)
\end{align}
and
\begin{align}\label{eq:expansion_2}
& W_{\frac{-2n_T-n+3}{2}, \frac{2n_T-n}{2}}\left(\frac{ x}{\phi}
\right) \nonumber\\
&  = \exp
\left(-\frac{x}{2\phi}\right)\left(\frac{x}{\phi}\right)^{\frac{-2n_T-n+3}{2}}
\left( 1+ O \left(\frac{1}{ x} \right) \right).
\end{align}
Substituting (\ref{eq:expansion_1}) and (\ref{eq:expansion_2}) into
(\ref{eq:cdf_lower}) and (\ref{eq:cdf_upper}) respectively, after
some manipulations we obtain (\ref{eq:tail_lower}) and
(\ref{eq:tail_upper}).

\section{Proof of \emph{Lemma}
\ref{lemma:max_gamma_bound}}\label{proof:lemma_max_gamma_bound}

Applying the results in \cite[Appendix I]{Mohammad08} to the
distribution of $\tilde{\gamma}_{k}(n)$ in (\ref{eq:cdf_lower}) with
$a_{K,n}= \phi$ and $b_{K,n} = \phi \log \big( \frac{c_n
K}{\phi^{n_T-1}}\big) - \phi (n_T+n-2) \log\log\big(\frac{c_n
K}{\phi^{n_T-1}}\big)$, we have the asymptotic results
(\ref{eq:asym_distribution}) at the top of next page.
\begin{figure*}[!t]
\begin{align}\label{eq:asym_distribution}
&\lim_{K \rightarrow \infty} K \left( 1 -
F_{\tilde{{\gamma}}(n)}(a_{K,n}x+b_{K,n})\right) = \lim_{K
\rightarrow \infty} \frac{c_n \phi^{n-1} K \exp
\left(-\frac{a_{K,n}x+b_{K,n}}{\phi}\right)}
{\left(a_{K,n}x+b_{K,n} \right)^{n_T+n-2}} \nonumber\\
& \hspace{3mm} = \lim_{K \rightarrow \infty} \frac{ c_n \phi^{n-1} K
\exp \left[-x - \log \left(\frac{c_n K}{\phi^{n_T-1}} \right)
+(n_T+n-2) \log\log \left(\frac{c_n K}{\phi^{n_T-1}} \right)
\right]}{\left[ \phi x + \phi \log \big( \frac{c_n
K}{\phi^{n_T-1}}\big)- \phi (n_T+n-2) \log\log\big(\frac{c_n
K}{\phi^{n_T -1}}\big)
\right]^{n_T+n-2}}\nonumber\\
& \hspace{3mm} =  \lim_{K \rightarrow \infty} \frac{ e^{-x}
\left(\log \left(\frac{c_n K}{\phi^{n_T-1}} \right)
\right)^{n_T+n-2}}{\left[
 x +  \log \big( \frac{c_n K}{\phi^{n_T-1}}\big)-
(n_T+n-2) \log\log\big(\frac{c_n K}{\phi^{n_T -1}}\big)
\right]^{n_T+n-2}}\nonumber\\
& \hspace{3mm} =e^{-x}.
\end{align}
\hrulefill
\end{figure*}
Similarly, we can show that $\lim_{K \rightarrow \infty} K \left( 1
- F_{\bar{{\gamma}}(n)}(a_{K,n}x+b_{K,n})\right) =e^{-x}$. Borrowing
the language from extreme value theory of order statistics,
$F_{{{\tilde{\gamma}}}(n)}(x)$ and $F_{\bar{{\gamma}}(n)}(x)$ are in
the domain of attraction of type (III) limiting distribution
\cite[\emph{Theorem} 4]{Mohammad08}. The final results can be
obtained by utilizing a similar method used in \cite[\emph{Lemma}
7]{Mohammad08}.

\section{Proof of Lemma
\ref{lemma:sinr_up_low}}\label{proof:lemma_sinr_up_low} Recall that
$F_{{\tilde{\gamma}}_{\mathcal{S}(n)}}(x) \leq
F_{{\gamma}_{\mathcal{S}(n)}}(x) \leq
F_{{\bar{\gamma}}_{\mathcal{S}(n)}}(x) $. For
$\gamma_{\mathcal{S}(n)}, n \in \{ 2,\ldots, M \}$, and large $K$,
with (\ref{eq:extre_gamma_bar}), $\text{Pr}\{ \chi_n - \phi \log\log
\sqrt K \leq \gamma_{\mathcal{S}(n)} \} \geq \text{Pr}\{ \chi_n -
\log\log \sqrt K \leq \bar{\gamma}_{\mathcal{S}(n)}\} \geq 1-
O\bigg( \frac{1}{\log K} \bigg)$. Similarly, with
(\ref{eq:extre_gamma_til}) we have $\text{Pr}\{
\gamma_{\mathcal{S}(n)} \leq \chi_n + \phi \log\log \sqrt K \} \geq
\text{Pr}\{ \tilde{\gamma}_{\mathcal{S}(n)}\leq  \chi_n + \log\log
\sqrt K \} \geq 1- O\bigg( \frac{1}{\log K} \bigg)$. Thus,
(\ref{eq:sinr_up_low}) holds.
For the case $n = 1$, the asymptotic distribution of
$\gamma_{\mathcal{S}(1)}$ is obtained in \cite[\emph{Theorem}
1]{Yoo_limited} as
\begin{eqnarray}\label{eq:sinr1_up_low}
& \hspace{-4cm}\text{Pr}\{ \chi_1 - \phi \log\log \sqrt{K} \leq
\gamma_{\mathcal{S}(1)} \nonumber\\
& \leq \chi_1 + \phi \log\log \sqrt{K}\} \geq 1 -
O\bigg(\frac{1}{\log K}\bigg).
\end{eqnarray}
The lemma is proved by combining (\ref{eq:sinr_up_low}) and
(\ref{eq:sinr1_up_low}).


\section{Proof of Theorem
\ref{theorem:sum_rate_THP_feedback}}\label{proof:theorem_sum_rate_THP_feedback}
Using (\ref{eq:sinr_up_low}) we can obtain $ \text{Pr}\bigg\{
\frac{\log_2 (1 +  \chi_n - \phi \log\log\sqrt{K})}{ \log_2{[\varrho
\log K]} }  \leq \frac{ \log_2 (1+
{\gamma}_{\mathcal{S}(n)})}{\log_2[\varrho \log K]} \leq
\frac{\log_2 (1 + \chi_{n} + \phi\log\log\sqrt{K})}{ \log_2{[\varrho
\log K]} } \bigg\} \geq 1- O \bigg(\frac{1}{\log K} \bigg)$.
Substituting (\ref{eq:chi_n}) and (\ref{eq:chi_1}) and letting $K
\rightarrow \infty$, the left-hand side and the right-hand side
inequality within $\text{Pr} \{\cdot \}$ converge to the same value.
Thus, $\lim_{K \rightarrow \infty} \frac{ \log_2 (1+
{\gamma}_{\mathcal{S}(n)})}{\log_2[ \varrho \log K]} = 1 $ with
probability 1, and (\ref{eq:converge_1}) holds. To establish
(\ref{eq:converge_2}), we employ the upper bound on
$\mathbb{E}\{R_{\text{BC}}\}$ derived in \cite{Sharif05}, which is
given as $\mathbb{E} \{R_{\text{BC}}\} \leq n_T \log_2 \big( 1 +
\varrho (\log K + O(\log \log K)) \big)$.
From \emph{Lemma} \ref{lemma:sinr_up_low}, we have $\text{Pr}\bigg\{
\log_2 (1+ {\gamma}_{\mathcal{S}(n)}) \geq \log_2 (1 + \chi_{n} -
\phi\log\log\sqrt{K}) \bigg\} \geq 1- O \bigg(\frac{1}{\log K}
\bigg)$. Thus, we have (\ref{eq:order_difference_THP}) at the top of
this page,
\begin{figure*}[!t]
\begin{eqnarray}\label{eq:order_difference_THP}
&& \hspace{-0.4cm}\mathbb{E} \{ R_{\text{BC}}\} - \mathbb{E} \{
R_{\text{G-THP-Q}} \} \nonumber\\
&& \leq  n_T \log_2 \big( 1 + \varrho (\log K + O(\log \log K))
\big) - \bigg( 1 - O \bigg(\frac{1}{\log K } \bigg)
\bigg)\sum_{n=1}^{n_T} \log_2 \big(1 + \chi_{n} -
\phi\log\log\sqrt{K}
\big) \nonumber\\
&&= \sum_{n=1}^{n_T} \log_2 \big(\frac{1 + \varrho (\log K + O(\log
\log K))} {1 + \chi_{n} - \phi\log\log\sqrt{K}}\big)+ O
\bigg(\frac{1}{\log K } \bigg) \sum_{n=1}^{n_T} \log_2 \big(1
+ \chi_{n} - \phi\log\log\sqrt{K} \big)\nonumber \\
&& \sim \sum_{n=1}^{n_T} \log_2\bigg(1 + \frac{\frac{P}{M n_T}\log
K+ O \left(\log\log{K}\right)+ \phi \log \left(
\frac{\phi^{n_T-1}}{c_n}\right) }{1 + \frac{ (M-1)P}{M n_T}\log K +O
\left(\log\log{K}\right)}\bigg)  +  O
\bigg(\frac{1}{\log K } \bigg)  O(\log \log K ) \nonumber\\
\label{eq:difference} && \sim  \frac{n_T}{\log2} \log \bigg(1+
\frac{1}{M-1} + O \left(\frac{\log\log K}{\log K} \bigg)  + O
\left(\frac{1}{\log K} \right) \right) + O \bigg(\frac{\log \log
K}{\log K}
\bigg) . 
\end{eqnarray}
\hrulefill
\end{figure*}
where $x\sim y$ means $\lim_{K \rightarrow \infty} x/y = 1$. Thus
(\ref{eq:converge_2}) is proved. Equation (\ref{eq:converge_3}) is
proved by combining (\ref{eq:difference}) and the fact that
$\log(1+x) \approx x$ for $x \ll 1$.

\section{Proof of \emph{Lemma}
\ref{lemma:cdfn_highSNR}}\label{proof:lemma_cdfn_highSNR} When $x
\geq \frac{1}{\delta} - 1$ and $\frac{x}{y} \geq \frac{1}{\delta}
-1$, $y \leq \frac{x}{\frac{1}{\delta} - 1} $ and
$\frac{x}{\frac{1}{\delta} - 1} \geq 1$. In addition, $\omega_k(n) $
and $Z
\defeq \frac{ \cos^2\theta_k}{ \sin^2\theta_k} $ are independent, thus
the c.d.f. of $ \hat{\gamma}_{k}(n)$ can be obtained as
\begin{align}\label{eq:cdfn_highSNR_derive1}
& F_{\hat{\gamma}(n)}(x) \nonumber \\ &= \int_{0}^{1} F_{Z}
\bigg(\frac{x}{y} \bigg)
f_{\omega(n)} (y) {\rm d} y ,\nonumber\\
&= \int_{0}^{1} \left( 1-
\frac{2^B}{\left(\frac{x}{y}+1\right)^{n_T-1}} \right)
f_{\omega(n)} (y) {\rm d} y\nonumber\\
&= 1- \int_{0}^{1} \frac{2^B}{\left(\frac{x}{y}+1\right)^{n_T-1}}
\frac{1}{\beta(n_T-n+1,n-1)} \nonumber\\ & \hspace{1cm}\times y^{n_T-n}(1- y)^{n-2} \nonumber\\
& =1- \frac{2^B}{\beta(n_T-n+1,n-1)} \int_{0}^{1} \frac{y^{2n_T -n
-1}  y^{n-2}}{(x+y)^{n_T-1}}{\rm d} y \nonumber \\
& = 1- \frac{2^B \beta(2n_T -n, n-1)}{\beta(n_T-n+1,n-1)}
\frac{1}{x^{n_T-1}} \nonumber \\ & \hspace{1cm}\times
{}_2F_1\left(n_T-1,2n_T-n;2n_T-1; -\frac{1}{x} \right),
\end{align}
where to obtain the last line we have used
\cite[3.197.3]{Gradshteyn2000}.

\section{Proof of Lemma
\ref{lemma:extream_cdfn_highSNR}}\label{proof:lemma_extream_cdfn_highSNR}
First, for $n=1$ the extremal distribution in
(\ref{eq:extre_cdf1_highSNR}) has been obtained in
\cite{Yoo_limited}. By applying (\ref{eq:tail_cdfn_highSNR}) with
the change of variables $\acute{\gamma}(n)=
\left(\frac{2^B}{d_n}\right)^{\frac{1}{n_T-1}}\hat{\gamma}(n) -1$,
we have
\begin{align}
 F_{\acute{\gamma}(n)} (x) & = F_{\hat{\gamma}(n)}\left(
\frac{x+1}{\left( \frac{2^B}{d_n}\right)^{\frac{1}{n_T-1}}}\right)
\nonumber \\ & = 1 - \frac{2^B}{(1+x)^{n_T-1}} + O\left(
\frac{1}{x}\right),
\end{align}
which is the same as the tail distribution of
$F_{\hat{\gamma}(1)}(x)$. Using (\ref{eq:extre_cdf1_highSNR}) we
have
\begin{align}
&\text{Pr} \bigg\{ \left(\frac{2^B
K}{\log\sqrt{K}}\right)^{\frac{1}{n_T-1}}  - 1 \leq \left(
\frac{2^B}{d_n}\right)^{\frac{1}{n_T-1}}
\hat{\gamma}_{\mathcal{S}(n)} -1 \nonumber\\ & \hspace{0.5cm} \leq
\left(2^B K\log\sqrt{K}\right)^{\frac{1}{n_T-1}}-1 \bigg\} \geq 1- O
\left(\frac{1}{\log K}\right).
\end{align}
Thus (\ref{eq:extre_cdfn_highSNR}) follows.

\bibliographystyle{IEEEtran}
\bibliography{sunliang_bib}

\end{document}